%% file: FastSolver.tex
\newif\ifanonym
\anonymtrue
\anonymfalse

\documentclass[11pt]{article}
\usepackage[latin9]{inputenc}
\usepackage{float}
\usepackage{amsmath}
\usepackage{amsthm}
\usepackage{amssymb}
\usepackage{setspace}
\usepackage{xcolor}

\usepackage{ifpdf}
\ifpdf    %
\usepackage{hyperref}
\else    %
\usepackage[hypertex]{hyperref}
\fi
\hypersetup{colorlinks=true,allcolors=blue}

\makeatletter

\floatstyle{ruled}
\newfloat{algorithm}{tbp}{loa}
\providecommand{\algorithmname}{Algorithm}
\floatname{algorithm}{\protect\algorithmname}

\theoremstyle{plain}
\newtheorem{thm}{\protect\theoremname}[section]
  \newtheorem{lemma}[thm]{\protect\lemmaname}
  \newtheorem{prop}[thm]{\protect\propname}
  \newtheorem{fact}[thm]{\protect\factname}
  \newtheorem{claim}[thm]{\protect\claimname}
\usepackage{braket}
\usepackage{amsmath}
\usepackage{fullpage}

\makeatother

\providecommand{\claimname}{Claim}
\providecommand{\factname}{Fact}
\providecommand{\lemmaname}{Lemma}
\providecommand{\propname}{Proposition}
\providecommand{\theoremname}{Theorem}

\providecommand{\eqdef}{\stackrel{{\rm def}}{=}}

\def\compactify{\itemsep=0pt \topsep=0pt \partopsep=0pt \parsep=0pt}

\newcommand{\vecone}{\vec 1}
\newcommand{\allones}{\vecone}
\DeclareMathOperator*{\EX}{{\mathbb E}}
\DeclareMathOperator*{\E}{{\mathbb E}}

\DeclareMathOperator{\diag}{diag}
\DeclareMathOperator{\poly}{poly}
\DeclareMathOperator{\polylog}{polylog}

\DeclareMathOperator{\supp}{supp}
\DeclareMathOperator{\nnz}{nnz}

\DeclareMathOperator{\sgn}{sgn}
\DeclareMathOperator{\Reff}{R_{\textrm{eff}}}
\DeclareMathOperator{\hatReff}{\hat{R}_{\textrm{eff}}}
\DeclareMathOperator{\hatw}{{\hat w}}
\DeclareMathOperator{\barkappa}{{\bar\kappa}}
\newcommand{\transpose}{\mathsf{T}}
\newcommand{\tran}{\transpose}

\newcommand{\R}{\mathbb{R}}
\newcommand{\N}{\mathbb{N}}
\newcommand\norm[1]{\Vert {#1}\Vert}
\newcommand\abs[1]{\vert {#1}\vert}
\newcommand\card[1]{\vert {#1}\vert}

\newcommand{\maxx}[1]{\max\{{#1}\}}
\newcommand\tuple[1]{\langle {#1}\rangle}
\providecommand{\eqdef}{\coloneqq}
\providecommand{\set}[1]{{\{#1\}}}

\newcommand{\eps}{\epsilon}
\newcommand{\tO}{\tilde{O}}

\begin{document}

\title{On Solving Linear Systems in Sublinear Time}
\ifanonym
\author{}
\else
\author{Alexandr Andoni%
\thanks{Columbia University.
Email: \texttt{andoni@cs.columbia.edu}.}
\and Robert Krauthgamer%
\thanks{Weizmann Institute of Science.
This work was partially supported by the Israel Science Foundation grant \#897/13, by a Minerva Foundation grant, and a Google Faculty Research Award.
Part of this work was done while was visiting the Simons Institute for the Theory of Computing. 
Email: \texttt{\{yosef.pogrow,robert.krauthgamer\}@weizmann.ac.il}.
}
\and Yosef Pogrow\footnotemark[2]
}
\fi

\maketitle

\begin{abstract}
We study \emph{sublinear} algorithms that solve linear systems locally.  
In the classical version of this problem the input is a
matrix $S\in \R^{n\times n}$ and a vector $b\in\R^n$ in the range of
$S$, and the goal is to output $x\in \R^n$ satisfying $Sx=b$.  
For the case when the matrix $S$ is symmetric diagonally dominant (SDD), 
the breakthrough algorithm of Spielman and Teng [STOC 2004] approximately
solves this problem in near-linear time (in the input size which is
the number of non-zeros in $S$), and subsequent papers have further
simplified, improved, and generalized the algorithms for this setting.

Here we focus on computing one (or a few) coordinates of $x$, 
which potentially allows for sublinear algorithms.  Formally, given an
index $u\in [n]$ together with $S$ and $b$ as above, the goal is to
output an approximation $\hat{x}_u$ for $x^*_u$, where $x^*$ is a
fixed solution to $Sx=b$.

Our results show that there is a qualitative gap between SDD matrices 
and the more general class of positive semidefinite (PSD) matrices.
For SDD matrices, we develop an algorithm that
approximates a single coordinate $x_{u}$ in time that is
polylogarithmic in $n$, provided that $S$ is sparse and has a small
condition number (e.g., Laplacian of an expander graph).  The
approximation guarantee is additive $| \hat{x}_u-x^*_u | \le \epsilon
\| x^* \|_\infty$ for accuracy parameter $\epsilon>0$.  We further
prove that the condition-number assumption is necessary and tight.

In contrast to the SDD matrices, we prove that for certain PSD 
matrices $S$, the running time must be at least polynomial in $n$. 
This holds even when one wants to obtain the same additive approximation,
and $S$ has bounded sparsity and condition number.
\end{abstract}

\thispagestyle{empty}
\newpage
\setcounter{page}{1}

\input{intro}

\input{laplacian}

\input{psdLowerBound}
\input{lowerBoundK2}

\input{sddSolver}

\subsubsection*{Acknowledgments}
The authors thank anonymous reviewers for suggesting additional
relevant references. 

{\small 
\bibliographystyle{alphaurlinit}
\bibliography{robi}
}

\end{document}

%% file: intro.tex
\section{Introduction}

Solving linear systems is a fundamental problem in many areas. 
A basic version of the problem has as input 
a matrix $A\in\mathbb{R}^{n\times n}$ and a vector $b\in\mathbb{R}^{n}$, 
and the goal is to find $x\in\mathbb{R}^{n}$ such that $Ax=b$. 
The fastest known algorithm
for general $A$ is by a reduction to matrix multiplication, and
takes $O(n^{\omega})$ time, where $\omega<2.372$ \cite{Gal14} is
the matrix multiplication exponent.
When $A$ is sparse, one can do better (by applying the conjugate gradient method to the equivalent positive semidefinite (PSD) system $A^\tran Ax=A^\tran b$, see
for example \cite{Spi10}), namely, $O(mn)$ time where $m=\nnz(A)$ is the number of non-zeros in $A$.
Note that this $O(mn)$ bound for exact solvers assumes exact arithmetic, and in practice, one seeks fast approximate solvers.

One interesting subclass of PSD matrices is that of symmetric
diagonally dominant (SDD) matrices.%
\footnote{A symmetric matrix $S\in\mathbb{R}^{n\times n}$ is called SDD 
if $S_{ii}\geq\sum_{j\neq  i}|S_{ij}|$ for all $i\in[n]$.
}
Many applications require solving
linear systems in SDD matrices, and most notably their subclass of
graph-Laplacian matrices, see e.g.~\cite{Spi10,Vishnoi13,CKMPPRX14}.
Solving SDD linear systems received a lot of attention in the past
decade after the breakthrough result by Spielman and Teng in 2004
\cite{ST04a}, showing that a linear system in SDD matrix $S$ can be
solved approximately in near-linear time
$O(m\mbox{log}^{O(1)}n\log\frac{1}{\epsilon})$, where $m=\nnz(S)$ and
$\epsilon>0$ is an accuracy parameter. A series of improvements led to
the state-of-the-art SDD solver of Cohen et al. \cite{CKMPPRX14} that
runs in near-linear time $O(m\sqrt{\log n}(\log\log
n)^{O(1)}\log\frac{1}{\epsilon})$. Recent improvements extend to
connection Laplacians \cite{kyng2016sparsified}. Obtaining similar
results for all PSD matrices remains a major open question.

Motivated by fast linear-system solvers in alternative models, 
here we study which linear systems can be solved in {\em sublinear time}. 
We can hope for such sublinear times if only one
(or a few) coordinates of the solution $x\in \R^n$ are sought.
Formally, given a matrix $S\in\mathbb{R}^{n\times n}$, a vector
$b\in\mathbb{R}^{n}$, and an index $u\in[n]$, 
we want to approximate the coordinate $x_{u}$ of 
a solution $x\in\mathbb{R}^n$ to the linear system $Sx=b$
(assume for now the solution is unique),
and we want the running time to be \emph{sublinear in $n$}. 

Our main contribution is a {\em qualitative separation} between the
class of SDD matrices and the larger class of PSD matrices, as
follows.  For well-conditioned SDD matrices $S$, we develop a
(randomized) algorithm that approximates a single coordinate $x_{u}$
fast --- in $\polylog(n)$ time.  In contrast, for some
well-conditioned PSD (but not SDD) matrices $S$,
we show that the same task requires $n^{\Omega(1)}$ time.  In addition, we
justify the dependence on the condition number. 

Our study is partly motivated by the advent of {\em quantum}
algorithms that solve linear systems in sublinear time, which were
introduced in \cite{HHL09}, and subsequently improved in
\cite{ambainis2012,CKS17}, and meanwhile used for a number
of (quantum) machine learning algorithms (see, e.g., the survey
\cite{DBLP:journals/corr/abs-1802-08227}). In particular, \cite{HHL09}
consider the system $Ax=b$ given: (1) oracle access to entries of $A$
(including fast access to the $j$-th non-zero entry in the $i$-th
row), and (2) a fast black-box procedure to prepare a quantum state
$|b\rangle=\sum_i \tfrac{b_i|i\rangle}{\|\sum_i
  b_i|i\rangle\|}$. Then, if the matrix $A$ has condition number
$\kappa$, at most $d$ non-zeros per row/column, and $\|A\|=1$, their
quantum algorithm runs in time $\poly(\kappa, d, 1/\epsilon)$, and
outputs a quantum state $|\hat x\rangle$ within $\ell_2$-distance
$\eps$ from $|x\rangle=\tfrac{\sum_i x_i|i\rangle}{\|\sum_i
  x_i|i\rangle\|}$.  The runtime was later improved in
\cite{CKS17} to depend logarithmically on $1/\epsilon$.
(The original goal of~\cite{HHL09} was different --- to output a
``classical'' value, a linear combination of $|x\rangle$ --- and for
this goal the improved dependence on $1/\epsilon$ is not possible
unless $BQP=PP$.)  
These quantum sublinear-time algorithms raise the question 
whether there are analogues classical algorithms for the same problems; 
for example, a very recent success story is 
a classical algorithm~\cite{tang18-notQuantum} 
for a certain variant of recommendation systems, 
inspired by an earlier quantum algorithm~\cite{kerenidis2016quantum}.  
Our lower bound precludes a classical analogue to the aforementioned linear-system solver, 
which works for all matrices $A$ and in particular for PSD ones.

\paragraph{Problem Formulation.}
To formalize the problem, we need to address a common issue for linear
systems --- they may be underdetermined and thus have many solutions
$x$, which is a nuissance when solving for a single coordinate.
We require that the algorithm approximates a single solution $x^*$,
in the sense that invoking the algorithm with different indices $u\in[n]$
will output coordinates that are all consistent with one ``global'' solution. 
This formulation follows the concept of Local Computation Algorithms, 
see Section~\ref{sec:related}. 

Formally, given a matrix $S\in\mathbb{R}^{n\times n}$, 
a vector $b\in\mathbb{R}^{n}$ in the range (column space) of $S$, 
and an accuracy parameter $\epsilon>0$, 
there exists $x^*\in\mathbb{R}^n$ satisfying $Sx^*=b$, 
such that upon query $u\in[n]$ the (randomized) algorithm outputs $\hat{x}_{u}$ 
that satisfies
\begin{equation} \label{eq:formulation}
  \forall u\in[n], \qquad
  \Pr\Big[|\hat{x}_{u}-x^*_{u}|\leq\epsilon||x^*||_{\infty}\Big]
  \geq
  \tfrac34.
\end{equation}
This guarantee corresponds (modulo amplification of the success probability) 
to finding a solution $\hat x\in\R^n$ 
with $\norm{\hat x-x^*}_\infty\le \epsilon||x^*||_{\infty}$.
We remark that the guarantee in~\cite{ST04a} is different, 
that $||\hat{x}-x^*||_{S}\leq\epsilon||x^*||_{S}$ 
where $||y||_{S}\eqdef \sqrt{y^\tran Sy}$.

\paragraph{Basic Notation.}
Given a (possibly edge-weighted) undirected graph $G=(V,E)$, we assume 
for convenience $V=[n]$.  Its Laplacian is the matrix $L_G\eqdef
D-A\in\R^{n\times n}$, where $A$ is the (weighted) adjacency matrix of
$G$, and $D$ is the diagonal matrix of (weighted) degrees in $G$.  It
is well-known that all Laplacians are SDD matrices, which in turn are
always PSD.

The \emph{sparsity} of a matrix is the maximum number of non-zero
entries in a single row/column.
The \emph{condition number} of a PSD matrix $S$, denoted $\kappa(S)$,
is the ratio between its largest and smallest non-zero eigenvalues.%
\footnote{Our definition is in line with the standard one, for a
  general matrix $A$, which uses singular values instead of
  eigenvalues.  If $A$ is singular, one could alternatively define
  $\kappa(A)=\infty$, which would only make the problem simpler (it
  becomes easier to be linear in $\kappa$), see e.g.~\cite{Spi10}.  }
For example, for the Laplacian $L_G$ of a $d$-regular graph $G$, let
$\mu_1\le\ldots\le\mu_n$ denote its eigenvalues, then the condition
number is $\kappa(L_G) =\Theta(\frac{d}{\mu_2})$.  This follows from
two well-known facts, that $\mu_n\in[d,2d]$, and that $\mu_2>\mu_1=0$
if $G$ is connected ($\mu_2$ is called the spectral gap).  Throughout,
$||A||$ denotes the spectral norm of a matrix $A$, and $A^+$ denotes
the Moore-Penrose pseudo-inverse of $A$.%
\footnote{For a PSD matrix $A\in\mathbb{R}^{n\times n}$, let its
  eigen-decomposition be $A=\sum_{i=1}^n \lambda_i u_iu_i^\tran$, then the
  Moore-Penrose pseudo-inverse of $A$ is $A^+=\sum_{i:\lambda_i>0}
  \tfrac{1}{\lambda_i} u_iu_i^\tran$.  
}

\subsection{Our Results}
\label{sec:results}

Below we describe our results, which include both algorithms and lower bounds. 
First, we present a polylogarithmic-time algorithm 
for the simpler case of Laplacian matrices, 
and then we generalize it to all SDD matrices.  
We further prove two lower bounds, which show that our algorithms cannot be
substantially improved to handle more general inputs or to run faster. 
The first lower bound shows that general PSD matrices require polynomial time,
thereby showing a strong separation from the SDD case.  
The second one shows that our SDD algorithm's dependence 
on the condition number is necessary and in fact near-tight.

\paragraph{Algorithm for Laplacian matrices.}
We first present our simpler algorithm for linear systems in Laplacians 
with a bounded condition number.

\begin{thm}[Laplacian Solver, see Section~\ref{sec:Laplacian}]
\label{thm:Laplacian_solver_intro}
There exists a randomized algorithm, 
that given input $\left\langle G,b,u,\epsilon,\barkappa\right\rangle$, 
where 
\begin{itemize} \compactify
\item $G=(V,E)$ is a connected $d$-regular graph given as an adjacency list,
\item $b\in\mathbb{R}^{n}$ is in the range of $L_{G}$ (equivalently, orthogonal
to the all-ones vector),
\item $u\in[n]$, $\epsilon>0$, and
\item $\barkappa\ge1$ is an upper bound on the condition number $\kappa(L_G)$,
\end{itemize}
the algorithm outputs $\hat{x}_{u}\in\mathbb{R}$ with the following guarantee.
If $x^*$ satisfies $L_{G}x^*=b$ then 
\[
  \forall u\in[n], \qquad 
  \Pr\Big[ | \hat{x}_{u}-x^*_{u} |
    \leq\epsilon\cdot\|x^*\|_\infty\Big]
  \geq
  1-\tfrac{1}{s} , 
\]
and the algorithm runs in time $O(d\epsilon^{-2}s^{3}\log s)$,
for suitable $s=O(\barkappa\log(\epsilon^{-1}\barkappa n))$. 
\end{thm}

A few extensions of the theorem follow easily from our proof.
First, if the algorithm is given also an upper bound $B_{up}$ on $||b||_0$,  
then the expression for $s$ can be refined by replacing $n$ with $B_{up}\le n$.
Second, we can improve the runtime to $O(\epsilon^{-2}s^{3}\log
s)$ whenever the representation of $G$ allows to sample a uniformly
random neighbor of a vertex in constant time. 
Third, the algorithm has an (essentially) cubic dependence on the
condition number $\kappa(L_G)$,  
which can be improved to quadratic if we allow a preprocessing of $G$ 
(or, equivalently if we only count the number of probes into $b$). 
Later we show that this quadratic dependence is \emph{near-optimal}.

\paragraph{Algorithm for SDD matrices.}
We further design an algorithm for SDD matrices with bounded condition number. 
The formal statement appears in Theorem~\ref{thm:SDD_solver_intro}
and is a natural generalization of Theorem~\ref{thm:Laplacian_solver_intro}
with two differences.
One difference is that a natural solution to the system $Sx=b$ is $x=S^+b$, 
but our method requires $S$ to have normalized diagonal entries,
and thus we aim at another solution $x^*$, construed as follows. 
Define
\begin{equation} \label{eq:tildeS}
  D\eqdef \diag(S_{11},...,S_{nn})
  \quad \text{ and } \quad
  \tilde{S}\eqdef D^{-1/2}SD^{-1/2}, %
\end{equation}
then our linear system can be written as $\tilde{S}(D^{1/2}x)=D^{-1/2}b$,
which has a solution 
\begin{equation}
\label{eqn:normalizedSsol}
  x^* \eqdef D^{-1/2}\tilde{S}^{+}D^{-1/2}b ,
\end{equation}
that is expressed using the pseudo-inverse of $\tilde{S}$, 
rather than of $S$.

A second difference is that Theorem~\ref{thm:SDD_solver_intro} 
makes no assumptions about the multiplicity of 
the eigenvalue $0$ of $\tilde{S}$, 
e.g., if $S$ is a graph Laplacian, then the graph need not be connected.  
The assumptions needed to achieve a polylogarithmic time, 
beyond $\tilde S$ having a bounded condition number, 
are only that a random ``neighbor'' in the graph corresponding to $S$
can be sampled quickly, 
and that $\frac{\max_{i\in[n]}D_{ii}}{\min_{i\in[n]}D_{ii}}\leq\mbox{poly}(n)$,
which holds if $S$ has polynomially-bounded entries. 

\paragraph{Lower Bound for PSD matrices.}
Our first lower bound shows that the above guarantees 
cannot be obtained for a general PSD matrix, 
even if we are allowed to preprocess the matrix $S$, 
and only count probes into $b$.
The proof employs a PSD matrix $S$ that is invertible 
(i.e., positive definite), 
in which case the linear system $Sx=b$ has a unique solution $x=S^{-1}b$.

\begin{thm}[Lower Bound for PSD Systems, see Section~\ref{sec:lbPSD}]
\label{thm:lbPSD}
For every large enough $n$, 
there exists an invertible PSD matrix $S\in\R^{n\times n}$ 
with uniformly bounded sparsity $d=O(1)$ and condition number $\kappa(S)\le 3$,
and a distinguished index $u\in[n]$, with the following property.  
Every randomized algorithm that, given as input $b\in \R^n$,
outputs $\hat x_u$ satisfying
$$
  \Pr \Big[ |\hat x_u-x^*_u|\le \tfrac15 \|x^*\|_\infty \Big] 
  \ge 
  \tfrac67 ,
$$
where $x^*=S^{-1}b$,
must probe $n^{\Omega(1/d^2)}$ coordinates of $b$ (in the worst case).
\end{thm}

\paragraph{Dependence on Condition Number.} 
The second lower bound shows that our SDD algorithm has a {\em
  near-optimal} dependence on the condition number of $S$, even if we
are allowed to preprocess the matrix S, and only count probes into
$b$. The lower bound holds even for Laplacian matrices.

\begin{thm}[Lower Bound for Laplacian Systems, see Section~\ref{sec:lbK2}]
\label{thm:lbK2}
For every large enough $n$ and $k\le O(n^{1/2}/\log n)$,
there exist an unweighted graph $G=([n],E)$ with maximum degree $4$
and whose Laplacian $L_G$ has condition number $\kappa(L_G)=O(k)$, 
and an edge $(u,v)$ in $G$, which satisfy the following.  
Every randomized algorithm that, given input $b$ in the range of $L_G$,
succeeds with probability $2/3$ to approximate $x_u-x_v$
within additive error $\eps\|x^*\|$ for $\eps=\Theta(1/\log n)$, 
where $x^*\in\R^n$ is any solution to $L_G x=b$, 
must probe $\tilde\Omega(k^2)$ coordinates of $b$ (in the worst case).
\end{thm}

\paragraph{Applications.}
An example application of our algorithmic results 
is computing the effective resistance between a pair of vertices $u,v$ 
in a graph $G$ (given $u$,$v$ and $G$ as input). 
It is well known that the effective resistance, denoted $\Reff(u,v)$,
can be expressed as $x_{u}-x_{v}$, where $x$ solves $L_{G}x=e_{u}-e_{v}$.
The spectral-sparsification algorithm of Spielman and Srivastava~\cite{SS11} 
relies on a near-linear time algorithm (that they devise)
for approximating the effective resistances of all edges in $G$. 
For unweighted graphs, there is also a faster algorithm~\cite{Lee14} 
that runs in time $\tO(n)$, which is sublinear in the number of edges, 
and approximates effective resistances within a larger factor $\polylog(n)$. 
In a $d$-regular expander $G$, 
the effective resistance between every two vertices is $\Theta(1/d)$,
and our algorithm in Theorem~\ref{thm:Laplacian_solver_intro} 
can quickly compute an arbitrarily good approximation (factor $1+\epsilon$). 
Indeed, observe that we can use $B_{up}=2$, 
hence the runtime is $O(\frac{1}{\epsilon^2}\polylog\frac{1}{\epsilon})$, 
independently of $n$.
The additive accuracy is $\epsilon\cdot\max_{ij\in E(G)} |x_i-x_j|$;
in fact, each $x_i-x_j$ is the potential difference between $i$ and $j$
when a potential difference of $\Reff(u,v)$ is imposed between $u$ and $v$,
thus 
$\max_{ij\in E(G)} \abs{x_i-x_j} \leq {x_u-x_v} = \Reff(u,v)$,
and hence with high probability the output actually achieves 
a multiplicative guarantee $\hatReff(u,v)\in (1\pm\epsilon)\Reff(u,v)$.

\subsection{Technical Outline}
\label{subsec:techniques}

\paragraph{Algorithms.}
Our basic technique relies on a classic idea of von Neumann and
Ulam~\cite{FL50,Was52} for estimating a matrix inverse by a power series; 
see Section~\ref{sec:related} for a discussion of related work.
Our starting point is the identity
\[
  \forall X\in\R^{n\times n}, \norm{X}<1, 
  \qquad 
  (I-X)^{-1}=\sum_{t=0}^{\infty}X^{t}. 
\]
(Recall that $||X||$ denotes the spectral norm of a matrix $X$.)
Now given a Laplacian $L=L_{G}$ of a $d$-regular graph $G$, 
observe that $\frac{1}{d}L=I-\frac{1}{d}A$, 
where $A$ is the adjacency matrix of $G$.  
Assume for a moment that $||\frac{1}{d}A||<1$; 
then by the above identity, 
$(\frac{1}{d}L)^{-1}=(I-\frac{1}{d}A)^{-1}=\sum_{t=0}^{\infty}(\frac{1}{d}A)^{t}$,
and the solution of the linear system $Lx=b$ would be
$x^*=L^{-1}b=\frac{1}{d}\sum_{t=0}^{\infty}(\frac{1}{d}A)^{t}b$.  
The point is that the summands decay exponentially because
$||(\frac{1}{d}A)^{t}b||_2\leq||(\frac{1}{d}A)^{t}||\cdot||b||_2
\leq||(\frac{1}{d}A)||^{t}\cdot||b||_2$.  
Therefore, we can estimate $x^*_{u}$ using the first $t_{0}$ terms, i.e.,
$\hat{x}_{u}=e_{u}^\tran \frac{1}{d}\sum_{t=0}^{t_{0}}(\frac{1}{d}A)^{t}b$, 
where $t_{0}$ is logarithmic (with base $\norm{\frac{1}{d}A}^{-1}>1$).
In order to compute each term $e_{u}^\tran \frac{1}{d}(\frac{1}{d}A)^{t}b$, 
observe that $e_{u}^\tran (\frac{1}{d}A)^{t}e_{w}$ is exactly the probability 
that a random walk of length $t$ starting at $u$ will end at vertex $w$. 
Thus, if we perform a random walk of length $t$ starting at $u$,
and let $z$ be its (random) end vertex, then 
\[
  \EX_z[b_{z}]
  =\sum_{w\in V} e_{u}^\tran (\tfrac{1}{d}A)^{t}e_{w}b_{w}
  =e_{u}^\tran (\tfrac{1}{d}A)^{t}b.
\]
If we perform several random walks (specifically,
$\poly(t_{0},\frac{1}{\epsilon})$ walk suffice), average the
resulting $b_{z}$'s, and then multiply by $\frac{1}{d}$, then with
high probability, we will obtain a good approximation to
$e_{u}^\tran \frac{1}{d}(\frac{1}{d}A)^{t}b$.

As a matter of fact, we have a non-strict inequality $||\frac{1}{d}A|| \le 1$,
because of the all-ones vector $\allones\in\R^n$.
Nevertheless, we can still get a meaningful result 
if all eigenvalues of $A$ except for the largest one are smaller than $d$ 
(equivalently, the graph $G$ is connected).
First, we get rid of any negative eigenvalues by the standard trick of
considering $(dI+A)/2$ instead of $A$, which is equivalent to adding
$d$ self-loops at every vertex.  
We may assume $b$ is orthogonal to $\allones$ 
(otherwise the linear system has no solution), 
hence the linear system $Lx=b$ has infinitely many solutions, 
and since $L$ is PSD, we can aim to estimate the specific solution 
$x^* \eqdef L^{+}b$ by $\frac{1}{d}\sum_{t=0}^{t_{0}}(\frac{1}{d}A)^{t}b$.  
Indeed, the idealized analysis above still applies by restricting all our
calculations to the subspace orthogonal to $\allones$.
This is carried out in Theorem~\ref{thm:Laplacian_solver_intro}. 

To generalize the above approach to SDD matrices, we face three issues.  
First, due to the irregularity of general SDD matrices, 
it is harder to properly define the equivalent random walk matrix.  
We resolve this by normalizing the SDD matrix $S$ 
into $\tilde{S}$ defined in~\eqref{eq:tildeS}, 
and solving the equivalent (normalized) system $\tilde{S}(D^{1/2}x)=D^{-1/2}b$.  
Second, general SDD matrices can have {\em positive} off-diagonal elements, 
in constrast to the Laplacians. 
To address this, we interpret such entries as negative-weight edges, 
and employ random walks that ``remember'' the signs of the traversed edges. 
Third, diagonal elements may strictly dominate their row, which we address by
terminating the random walk early with some positive probability.

\paragraph{Lower Bound: Polynomial Time for PSD Matrices.}
We first discuss our lower bound for PSD matrices, which is one of the
main contributions of our work. The goal is to exhibit a matrix $S$
for which estimating a coordinate $x^*_u$ of the solution $x^*=S^{-1}b$ 
requires $n^{\Omega(1)}$ probes into the input $b$. 

Without the sparsity constraint on $S$, one can deduce such a lower bound
via a reduction to the communication complexity of the \emph{Vector in
  Subspace Problem} (VSP), in which Alice has an $n/2$-dimensional
subspace $H\subset \R^n$, Bob has a vector $b\in\R^n$, and their goal
is to determine whether $b\in H$ or $b\in H^\perp$.  The randomized
communication complexity of this promise problem is between
$\Omega(n^{1/3})$ \cite{KR11b} and $O(\sqrt{n})$ \cite{Raz99} (while
for quantum communication it is $O(\log n)$).  To reduce this problem
to linear-system solvers, let $P_H\in R^{n\times n}$ be the projection
operator onto the subspace $H$, and set $S=I + P_H$.  Consider the
system $Sx=b$, and note that Alice knows $S$ and Bob has $b$.
It is easy to see that the unique solution $x^*$ is either $b$ or
$\tfrac12 b$, depending on whether $b\in H^\perp$ or $b\in H$.  Alice
and Bob could use a solver that makes few probes to $b$, as follows.
Bob would pick an index $u\in[n]$ that maximizes $\abs{b_u}$
(and thus also $\abs{x_u}$), and send it to Alice.  
She would then apply the solver, 
receiving from Bob only a few entries of $b$, 
to estimate $x_u$ within 
additive error $\tfrac12 \norm{x}_\infty$, %
which suffices to distinguish the two cases.  This matrix $S$ is PSD
with condition number $\kappa(S)\le 2$.  However it is dense.

We thus revert to a different approach of proving it from basic
principles.  Our high-level idea is to take a $2d$-regular expander
and assign its edges with signs ($\pm 1$) that are random but balanced everywhere (namely, at every vertex, the incident edges are split evenly between positive and negative).  The signed adjacency matrix
$A\in\set{-1,0,+1}^{n\times n}$ should have spectral norm $\mu\eqdef \norm{A}
=O(\sqrt{d})$, and then instead of the (signed) Laplacian $L=(2d)I-A$, 
we consider $S=2\mu I - A$, 
which is PSD with condition number $\kappa(S)\le 3$, as
well as invertible and sparse.  Now following similar arguments as in
our algorithm, we can write $S^{-1}$ as a power series of the matrix
$A$, and express coordinate $x^*_u$ of the solution $x^*=S^{-1}b$ via
$\EX_z[b_z]$ where $z$ is the (random) end vertex of a random walk 
that starts at $u$ and its length is bounded by some $t_0$
(performed in the ``signed'' graph corresponding to $A$).  
Now if the graph around $u$ looks
like a tree (e.g., it has high girth), then not-too-long walks are
highly symmetric and easy to count.  We now let $b_v$ be non-zero
only at vertices $v$ at distance exactly $t_0$ from $u$, and for these
vertices $b_v$ is set to $+1$ or $-1$ at random but with a small bias $\delta$ towards one of the values.  Some calculations show that $\sgn(\EX_z[b_z])$, and
consequently $\sgn(x^*_u)$, will be according to our bias (with high
probability), however discovering this $\sgn(x^*_u)$ via probes to $b$ 
is essentially the problem of learning a biased coin, which requires
$\Omega(\delta^{-2})$ coin observations. An additional technical
obstacle is to prove that the solution $x^*$ has a small $\ell_\infty$-norm, 
so that we can argue that an $\tfrac15 \norm{x^*}_\infty$-additive error 
to $x^*_u$ will not change its sign. Overall, we show we can set $t_0=\Omega(\log n)$ and $\delta \approx ((2d-1)^{t_0})^{-1/2}$, 
thus concluding that the algorithm must observe 
$\Omega(\delta^{-2})=n^{\Omega(1)}$ entries of $b$.

It is instructive to ask where in the above argument is it crucial to
have $\mu=O(\sqrt{d})$, because if it were also valid for $\mu=d$, in
which case the matrix $S=2\mu I - A$ is SDD, then it would contradict
our own algorithm for SDD matrices.  The answer is that $\mu\ll 2d$ is
required to bound $\norm{x^*}_\infty$, specifically in the analysis
that follows immediately after Eqn.~\eqref{eq:alpha}.

\paragraph{Lower Bound: Quadratic Dependence on Condition Number.}
We now outline the ideas to prove the $\tilde\Omega(\kappa^2)$ lower
bound even for Laplacian systems with condition number $\kappa$. First
we note that it is relatively straight-forward to prove that a {\em
  linear} dependence on the condition number is necessary. Namely,
consider a dumbbell graph (two 3-regular expanders connected by a
bridge edge $(u,v)$), for which we need to estimate $x^*_u-x^*_v$. For
$b$ defined as $b=e_i-e_j$, the value of $x^*_u-x^*_v$ will be
non-zero iff vertices $i,j$ are on the opposite sides of the
bridge. To determine the latter, one requires $\Omega(n)$ queries into
$b$. Since this graph has a condition number of $O(n)$, we obtain an
$\Omega(\kappa)$ lower bound.

The quadratic lower bound requires both a different graph and a
different distribution over $b$. We use the following graph $G$ with
condition number $O(k)$: take two 3-regular expanders and connect
them with $n/k$ edges (``bridge edges''). The vector $b\in
\{-1,+1\}^n$ will be {\em dense} and in particular it is either: 1)
balanced, i.e., $\sum b_i$ on each expander is zero, or 2) unbalanced,
i.e., each $b_i$ is chosen $\pm1$ with a bias $p\approx 1/k$ towards
$+1$ on the first expander, and towards $-1$ on the second one. Now,
as above, it is simple to prove that {\em on average over the bridge
  edges $(u,v)$}: 1) in the balanced case, the average of
$x^*_u-x^*_v$ must be zero, and 2) in the unbalanced case, the average
must be $\Omega(1)$. However, the main challenge is that the actual
values may differ from the average --- e.g., even in the balanced
case, each bridge edge $(u,v)$ will likely have non-zero value of
$x^*_u-x^*_v$. Nonetheless, we manage to prove an upper bound on the
maximum value of $|x^*_u-x^*_v|$ over all edges $(u,v)$ (as in the
previous lower bound, we need to bound $\|x^*\|_\infty$ as well). For
the latter, we need to again analyze $\E_z[b_z]$ where $z$ is the
endpoint of a random walk of some fixed length $i\ge 1$ starting from
$u$ in the graph $G$. Since the vector $b$ is not symmetric over the
graph $G$, a direct analysis seems hard --- instead we estimate
$\E_z[b_z]$ via a coupling of such walks in $G$ with random walks in
an expander, which is amenable to a direct analysis.

\subsection{Related Work}
\label{sec:related} 

The idea of approximating the inverse $(I-X)^{-1} = \sum_{t=0}^\infty
X^t$ (for $||X||<1$) by random walks dates back to von Neumann and
Ulam~\cite{FL50,Was52}.  While we approximate each power $X^t$ by
separate random walks of length $t$ and truncate the tail (powers
above some $t_0$), their method employs random walks whose length is
random and whose expectation gives exactly the infinite sum, 
achieved by assigning some probability to terminate the walk at each step,
and weighting the contributions of the walks accordingly (to correct the
expectation).

The idea of approximating a generalized inverse $L^{*}$ of $L=dI-A$ 
by the truncated series
$\frac{1}{d}\sum_{t=0}^{t_{0}}(\frac{1}{d}A)^{t}$ on directions that
are orthogonal to the all-ones vector was recently used by Doron, Le
Gall, and Ta-Shma \cite{DGT17} to show that $L^{*}$ can be
approximated in probabilistic log-space.  However, since they wanted
to output $L^{*}$ explicitly, they could not ignore the all-ones
direction and they needed to relate $L^{*}$ to
$\frac{1}{d}\sum_{t=0}^{\infty}(\frac{1}{d}A)^{t}$ by ``peeling off''
the all-ones direction, inverting using the infinite sum formula, and
then adding back the all-ones direction.

The idea of estimating powers of a normalized adjacency matrix
$\frac{1}{d}A$ (or more generally, a stochastic matrix) by performing
random walks is well known, and was used also in~\cite{DGT17} mentioned above,
and in~\cite{DST17}.  Chung and Simpson~\cite{CS15} used it in a context
that is related to ours, of solving a Laplacian system $L_Gx=b$ but
with a boundary condition, namely, a constraint that $x_i=b_i$ for all
$i$ in the support of $b$.  
Their algorithm solves for a subset of the coordinates $W\subseteq V$, 
i.e., it approximates $x|_{W}$ (the restriction of $x$ to coordinates in $W$) where $x$ solves $Lx=b$ under the boundary condition.  
They relate the solution $x$ to the Dirichlet heat-kernel
PageRank vector, which in turn is related to an infinite power series
of a transition matrix (specifically, to $ f^\tran e^{-t(I-P_W)} = e^{-t}
f^\tran \sum_{k=0}^{\infty}\frac{t^k}{k!}P_W^k $ where $P_W$ is the
transition matrix of the graph induced by $W$, $t\in\R$, and
$f\in\R^{|W|}$), and their algorithm uses random walks to approximate
the not-too-large powers of the transition matrix, proving that the
remainder of the infinite sum is small enough.

Recently, Shyamkumar, Banerjee and Lofgren \cite{SBL16} considered 
a related matrix-power problem,
where the input is a matrix $A\in \R^{n\times n}$, a power $\ell\in\N$, 
a vector $z\in\R^n$, and an index $u\in [n]$, 
and the goal is to compute coordinate $u$ of $A^\ell z$.
They devised for this problem a sublinear (in $\nnz(A)$) algorithm, 
under some bounded-norm conditions and assuming $u\in[n]$ is uniformly random.
Their algorithm relies, in part, on von Neumann and Ulam's technique 
of computing matrix powers using random walks, but of prescribed length.
It can be shown that approximately solving positive definite systems
for a particular coordinate is reducible to the matrix-power problem.%
\footnote{Let $Ax=b$ be a linear system where $A$ is positive definite.
Let $\lambda$ be the largest eigenvalue of $A$. Let $A'\eqdef\frac{1}{2\lambda}A$ and
$b'\eqdef\frac{1}{2\lambda}b$. Consider the equivalent system $(I-(I-A'))x=b'$.
As the eigenvalues of $A'$ are in $(0,1/2]$, the eigenvalues of
$I-A'$ are in $[1/2,1)$. Thus, the solution to the linear system is given by
$x=(I-(I-A'))^{-1}b'=\sum_{t=0}^\infty (I-A')^tb$. Therefore, we can
approximate $x_u$ by truncating the infinite sum at some $t_0$ and approximating
each power $t<t_0$ by the algorithm for the matrix-power problem. 
}
However, in contrast to our results, their expected runtime 
is polynomial in the input size, namely $\nnz(A)^{2/3}$,
and holds only for a random $u\in [n]$.

\paragraph{Comparison with PageRank.} 

An example application of our results is computing quickly the PageRank (defined in~\cite{BP98}) of a single node in an undirected $d$-regular graph.
Recall that the PageRank vector of an $n$-vertex graph with associated transition matrix $P$ is the solution to the linear system
$x=\frac{1-\alpha}{n}\allones + \alpha Px$,
where $0<\alpha<1$ is a given parameter. 
In personalized PageRank, one replaces $\frac{1}{n}\allones$ 
(the uniform distribution) with some $b\in\R^n$, e.g., a standard basis vector. 
Equivalently, $x$ solves the system $Sx=\frac{1-\alpha}{n}\allones$ 
where $S=I-\alpha P$ is an SDD matrix with $1$'s on the diagonal. 
As all eigenvalues of $P$ are of magnitude at most 1 
(recall $P$ is a transition matrix), all eigenvalues
of $I-\tilde{S}=I-S=\alpha P$ are of magnitude at most $\alpha$,
and the running time guaranteed by Theorem~\ref{thm:SDD_solver_intro} is
logarithmic (with base $\frac{2}{\alpha+1}$). 

Algorithms for the PageRank model were studied extensively.
In particular, the sublinear algorithms of~\cite{BPP18} approximate
the PageRank of a vertex using $\tO(n^{2/3})$ queries to an arbitrary graph,
or using $\tO((n\Delta)^{1/2})$ queries when the maximum degree is $\Delta$. 
Another example is the heavy-hitters algorithm of~\cite{BBCT14},
which reports all vertices whose approximate PageRank exceeds a threshold $T$
in sublinear time $\tO(1/\Delta)$, 
when PageRanks are viewed as probabilities and sum to $1$. 
Other work explores connections to other graph problems,
including for instances using PageRank algorithms 
to approximate effective resistances~\cite{CZ10},
the PageRank vector itself, and computing sparse cuts~\cite{ACL07}.

\paragraph{Local Algorithms.}
Our algorithms in Theorems~\ref{thm:Laplacian_solver_intro}
and~\ref{thm:SDD_solver_intro} are \emph{local}
in the sense that they query a small portion of their input,
usually around the input vertex, when viewed as graph algorithms.
Local algorithms for graph problems were studied in several contexts,
like graph partitioning~\cite{ST04a,AP09}, Web analysis~\cite{CGS04,ABCHMT08},
and distributed computing~\cite{Suo13}. 
Rubinfeld, Tamir, Vardi, and Xie~\cite{RTVX11} 
introduced a formal concept of \emph{Local Computation Algorithms}
that requires consistency between the local outputs of multiple executions
(namely, these local outputs must all agree with a single global solution). 
As explained earlier, our problem formulation~\eqref{eq:formulation} 
follows this consistency requirement. 

\subsection{Future Work} 
One may study alternative ways of defining the problem of 
solving a linear system in sublinear time, 
in particular if the algorithm can access $b$ in a different way. 
For example, 
similarly to assumptions and guarantees in \cite{tang18-notQuantum}, 
the goal may be to produce an $\ell_2$-sample from the solution $x$
(i.e., report a random index in $[n]$ such that 
the probability of each coordinate $i\in[n]$ is proportional to $x_i^2$) 
assuming oracle access to an $\ell_2$-sampler from $b\in\R^n$,
i.e., use an $\ell_2$-sampler for $b$ to construct an $\ell_2$-sampler for $x$. 
Another version of the problem may ask to produce heavy hitters in $x$, 
assuming, say,%
\footnote{This kind of oracle seems necessary even when $S=I$.} 
heavy hitters in $b$ (which may be
useful for the PageRank application). We leave these extensions as
interesting open questions, focusing here on the classical access
mode to $b$, via queries to its coordinates.

%% file: laplacian.tex
\section{Laplacian Solver (for Regular Graphs)}
\label{sec:Laplacian}

In this section we prove Theorem~\ref{thm:Laplacian_solver_intro}. 
The ensuing description deals mostly with a slightly simplified scenario, 
where the algorithm is given not one but two vertices $u,v\in [n]$,
and returns an approximation $\hat{\delta}_{u,v}$ to $x_{u}-x_{v}$
with a slightly different error bound, 
see Theorem~\ref{thm:Laplacian_solver} for the precise statement. 
We will then explain the modifications required
to prove Theorem~\ref{thm:Laplacian_solver_intro}
(which actually follows also from 
our more general Theorem~\ref{thm:SDD_solver_intro}).

Let $G=(V=[n],E)$ be a connected $d$-regular graph
with adjacency matrix $A\in\mathbb{R}^{n\times n}$.
Let the eigenvalues of $A$ be 
$d = \lambda_{1} > \lambda_{2} \geq\cdots\geq \lambda_{n}$,
and let their associated orthonormal eigenvectors be $u_{1},\ldots,u_{n}$.
Then $u_{1}=\frac{1}{\sqrt{n}}\cdot\allones\in\mathbb{R}^{n}$,
and we can write $A=U\Lambda U^\tran $ 
where $U=[u_{1}\, u_{2} \ldots u_{n}]$ is unitary and $\Lambda=\diag(\lambda_{1},...,\lambda_{n})$.
For $u,v\in[n]$, let $\chi_{u,v}\eqdef e_{u}-e_{v}$
where $e_{i}$ is the $i$-th standard basis vector. 
Then the Laplacian of $G$ is given by 
\[
  L
  \eqdef\sum_{uv\in E}\chi_{u,v}\chi_{u,v}^\tran  
  = dI-A
  = U(d I-\Lambda)U^\tran .
\]
Observe that $L$ does not depend on the orientation of each edge $uv$,
and that $\mu_2\eqdef d-\lambda_2$ is the smallest non-zero eigenvalue of $L$.
The Moore-Penrose pseudo-inverse of $L$ is 
\[
  L^{+}
  \eqdef 
  U\cdot \diag(0,(d-\lambda_{2})^{-1},\ldots,(d-\lambda_{n})^{-1})\cdot U^\tran .
\]
We assume henceforth that all eigenvalues of $A$ are non-negative. 
At the end of the proof, we will remove this assumption (by adding self-loops).

\medskip
The idea behind the next fact
is that $L=d(I-\frac{1}{d}A)$,
and $\frac{1}{d}A$ has norm strictly smaller than one when operating
on the subspace that is orthogonal to the all-ones vector, and hence,
the formula $(I-X)^{-1}=\sum_{t=0}^{\infty}X^{t}$ for $||X||<1$
is applicable for the span of $\{u_{2},...,u_{n}\}$.

\begin{fact}
\label{fact:inverse_power_series}For every $x\in\mathbb{R}^{n}$
that is orthogonal to the all-ones vector, $L^{+}x=\frac{1}{d}\sum_{t=0}^{\infty}(\frac{1}{d}A)^{t}x$.
\end{fact}

\begin{proof}
It suffices to prove the claim for each of $u_{2},\ldots,u_{n}$ as the
fact will then follow by linearity. Fix $i\in\{2,\ldots,n\}$. 
Then since $|\frac{\lambda_{i}}{d}|<1$, 
\[
  \sum_{t=0}^{\infty}\Big(\frac{1}{d}A\Big)^{t}u_{i}
  = \sum_{t=0}^{\infty}\Big(\frac{\lambda_{i}}{d}\Big)^{t}u_{i}
  = \frac{1}{1-\frac{\lambda_{i}}{d}}u_{i}=\frac{d}{d-\lambda_{i}}u_{i}
  = dL^{+}u_{i}.
\qedhere
\]
\end{proof}

We now describe an algorithm that on input $b\in\mathbb{R}^{n}$ that
is orthogonal to the all-ones vector, and two vertices $u\neq v\in[n]$,
returns an approximation $\hat{\delta}_{u,v}$ to $x_{u}-x_{v}$,
where $x$ solves $Lx=b$. As $G$ is
connected, the null space of $L$ is equal to span$\set{\allones}$
and hence $x_{u}-x_{v}$ is uniquely defined, and can be written as
$x_{u}-x_{v}=\chi_{u,v}^\tran L^{+}b$.

\begin{algorithm}
\caption{$\hat{\delta}_{u,v}=\mbox{SolveLinearLaplacian}(G,b,||b||_{0},u,v,\epsilon,d,\mu_2)$}
\label{alg:LaplacianSolver}

\begin{enumerate} \compactify
\item Set
\[
s=\frac{\log(2\sqrt{2}\epsilon^{-1}\frac{d}{\mu_2}\sqrt{||b||_{0}})}{\log(\frac{d}{d-\mu_2})},
\]
and $\ell=O((\frac{\epsilon}{4s})^{-2}\log s)$.
\item For $t=0,1,\ldots,s-1$ do

\begin{enumerate} \compactify
\item Perform $\ell$ independent random walks of length $t$ starting at
$u$, and let $u_{1}^{(t)},\ldots,u_{\ell}^{(t)}$ be the vertices at
which the random walks ended. Independently, perform $\ell$ independent
random walks of length $t$ starting at $v$, and let $v_{1}^{(t)},\ldots,v_{\ell}^{(t)}$
be the vertices at which the random walks ended.
\item Set $\hat{\delta}_{u,v}^{(t)}=\frac{1}{\ell}\sum_{i\in[\ell]}(b_{u_{i}^{(t)}}-b_{v_{i}^{(t)}})$.
\end{enumerate}
\item Return $\hat{\delta}_{u,v}=\frac{1}{d}\sum_{t=0}^{s-1}\hat{\delta}_{u,v}^{(t)}$.\end{enumerate}
\end{algorithm}

\begin{claim}
\label{claim:Laplacian_tail_bound}For $b$ that is orthogonal to
the all-ones vector, $|\chi_{u,v}^\tran L^{+}b-\chi_{u,v}^\tran \frac{1}{d}\sum_{t=0}^{s-1}(\frac{1}{d}A)^{t}b|\leq\frac{\epsilon}{2d}||b||_{\infty}$.
\end{claim}

\begin{proof}
Using Fact \ref{fact:inverse_power_series}, 
\[
\chi_{u,v}^\tran L^{+}b-\chi_{u,v}^\tran \frac{1}{d}\sum_{t=0}^{s-1}(\frac{1}{d}A)^{t}b=\chi_{u,v}^\tran \frac{1}{d}\sum_{t=s}^{\infty}(\frac{1}{d}A)^{t}b ,
\]
and thus 
\[
|\chi_{u,v}^\tran L^{+}b-\chi_{u,v}^\tran \frac{1}{d}\sum_{t=0}^{s-1}(\frac{1}{d}A)^{t}b|\leq||\chi_{u,v}^\tran ||_{2}\cdot||\frac{1}{d}\sum_{t=s}^{\infty}(\frac{1}{d}A)^{t}b||_{2}.
\]
 We know that $||\chi_{u,v}^\tran ||_{2}=\sqrt{2}$, so it remains to
bound $||\frac{1}{d}\sum_{t=s}^{\infty}(\frac{1}{d}A)^{t}b||_{2}$.
Decomposing $b=\sum_{i=2}^{n}c_{i}u_{i}$ we get that $\sum_{i=2}^{n}c_{i}^{2}=||b||_{2}^{2}$
and 
\[
\sum_{t=s}^{\infty}(\frac{1}{d}A)^{t}b=\sum_{i=2}^{n}c_{i}u_{i}\sum_{t=s}^{\infty}(\frac{\lambda_{i}}{d})^{t}=\sum_{i=2}^{n}\frac{(\frac{\lambda_{i}}{d})^{s}}{1-\frac{\lambda_{i}}{d}}c_{i}u_{i}=d\sum_{i=2}^{n}\frac{(\frac{\lambda_{i}}{d})^{s}}{d-\lambda_{i}}c_{i}u_{i}.
\]
Hence, 
\[
  ||\frac{1}{d}\sum_{t=s}^{\infty}(\frac{1}{d}A)^{t}b||_{2}^{2}
  =\sum_{i=2}^{n}\left(\frac{(\frac{\lambda_{i}}{d})^{s}}{d-\lambda_{i}}\right)^{2}c_{i}^{2}
    ||u_{i}||_{2}^{2}
    \leq\left(\frac{(\frac{\lambda_2}{d})^{s}}{d-\lambda_2}\right)^{2}\sum_{i=2}^{n}c_{i}^{2}
    =\left(\frac{(1-\frac{\mu_2}{d})^{s}}{\mu_2}\right)^{2}||b||_{2}^{2},
\]
where the first equality is because the $u_{i}$'s are orthogonal.
Altogether, 
\[
  |\chi_{u,v}^\tran L^{+}b-\chi_{u,v}^\tran \frac{1}{d}\sum_{t=0}^{s-1}(\frac{1}{d}A)^{t}b|
  \leq\sqrt{2}\frac{(1-\frac{\mu_2}{d})^{s}}{\mu_2}||b||_{2}
  \leq\sqrt{2}\frac{(1-\frac{\mu_2}{d})^{s}}{\mu_2}\sqrt{||b||_{0}}\cdot||b||_{\infty}
  =\frac{\epsilon}{2d}||b||_{\infty},
\]
as claimed.
\end{proof}

\begin{claim}
\label{claim:Laplacian_power_approximation}$\Pr\Big[|\hat{\delta}_{u,v}-\chi_{u,v}^\tran \frac{1}{d}\sum_{t=0}^{s-1}(\frac{1}{d}A)^{t}b|>\frac{\epsilon}{2d}||b||_{\infty}\Big]\leq\frac{1}{s}$.
\end{claim}

\begin{proof}
Observe that $e_{u}^\tran (\frac{1}{d}A)^{t}$ is a probability vector
over $V$, and $e_{u}^\tran (\frac{1}{d}A)^{t}e_{w}$ is exactly the
probability that a random walk of length $t$ starting at $u$ will
end at $w$. Thus, for every $t\in\{0,1,\ldots,s-1\}$ and $i\in[\ell]$,
we have
\[
  \mathbb{E}[b_{u_{i}^{(t)}}]
  =\sum_{w\in [n]} e_{u}^\tran (\frac{1}{d}A)^{t}e_{w}b_{w}
  =e_{u}^\tran (\frac{1}{d}A)^{t}b,
\]
and similarly $\mathbb{E}[b_{v_{i}^{(t)}}]=e_{v}^\tran (\frac{1}{d}A)^{t}b$.
By a union bound over Hoeffding bounds, with probability at least
$1-\frac{1}{s}$, for every $t\in\{0,1,\ldots,s-1\}$, we have $|\frac{1}{\ell}\sum_{i\in[\ell]}b_{u_{i}^{(t)}}-e_{u}^\tran (\frac{1}{d}A)^{t}b|\leq\frac{\epsilon}{4s}||b||_{\infty}$
and $|\frac{1}{\ell}\sum_{i\in[\ell]}b_{v_{i}^{(t)}}-e_{v}^\tran (\frac{1}{d}A)^{t}b|\leq\frac{\epsilon}{4s}||b||_{\infty}$.
Recalling that $\hat{\delta}_{u,v}=\frac{1}{d}\sum_{t=0}^{s-1}\frac{1}{\ell}\sum_{i\in[\ell]}(b_{u_{i}^{(t)}}-b_{v_{i}^{(t)}})$,
with probability at least $1-\frac{1}{s}$ we have $|\hat{\delta}_{u,v}-\chi_{u,v}^\tran \frac{1}{d}\sum_{t=0}^{s-1}(\frac{1}{d}A)^{t}b|\leq\frac{\epsilon}{2d}||b||_{\infty}$,
as claimed.
\end{proof}

Combining Claim \ref{claim:Laplacian_tail_bound} and Claim \ref{claim:Laplacian_power_approximation}
we get that (with probability $1-\frac{1}{s}$)
$|\hat{\delta}_{u,v}-\chi_{u,v}^\tran L^{+}b|\leq\frac{\epsilon}{d}||b||_{\infty}$.
Now, as $x$ solves $Lx=b$, 
for every $i\in[n]$ we have $\sum_{j\in N(i)}(x_{i}-x_{j})=b_{i}$
where $N(i)$ is the set of neighbors $\set{j:ij\in E}$, 
which implies that for some neighbor $j$ of $i$, 
it holds that $|x_{i}-x_{j}|\geq\frac{|b_{i}|}{d}$.
Therefore, $\max_{ij\in E}|x_{i}-x_{j}|\geq\frac{1}{d}||b||_{\infty}$.
We conclude that $|\hat{\delta}_{u,v}-\chi_{u,v}L^{+}b|\leq\epsilon\cdot\max_{ij\in E}|x_{i}-x_{j}|$.
We now turn to the running time of Algorithm~\ref{alg:LaplacianSolver}, 
which is dominated by the time it takes to perform the random walks.
There are $2s\cdot\ell$ random walks in total.
The random walks do not need to be independent for different values of $t$ 
(as we applied a union bound over the different $t$), 
we can extend, at each iteration $t$, the $2\ell$ respective
random walks constructed at iteration $t-1$ by an extra step in time
$O(d)$ (recall we assume $G$ is given as an adjacency list), 
obtaining a total runtime $O(s\cdot\ell\cdot d)=O(d\epsilon^{-2}s^{3}\log s)$.
To simplify the expression for $s$, we need the following bound.

\begin{fact}
\label{fact:inequality_e}
For all $\delta\in(0,1)$,
$\frac{1}{\ln(1-\delta)^{-1}}\leq \frac{1}{\delta}$.
\end{fact}
\begin{proof}
We need to show that $\delta \leq \ln(1-\delta)^{-1}$,
or equivalently, $e^{-\delta} \geq 1-\delta$,
which is well known. 
\end{proof}

Applying Fact~\ref{fact:inequality_e} to $\delta=\frac{\mu_2}{d}$, we have $s\leq\frac{d}{\mu_2}\log(2\sqrt{2}\epsilon^{-1}\frac{d}{\mu_2}\sqrt{||b||_0})$,
and conclude the following.

\begin{thm}
\label{thm:Laplacian_solver}
Given an adjacency list of a $d$-regular $n$-vertex graph $G$, 
a vector $b\in\mathbb{R}^{n}$ that is orthogonal to the all-ones vector, 
vertices $u,v\in[n]$, and scalars $||b||_{0}$, $\epsilon>0$, 
and $\mu_2=d-\lambda_2>0$, 
Algorithm~\ref{alg:LaplacianSolver} outputs $\hat{\delta}_{u,v}\in\mathbb{R}$ 
satisfying 
\[
  \Pr\Big[ |\hat{\delta}_{u,v}-\chi_{u,v}^\tran L^{+}b|\leq\epsilon\cdot\max_{ij\in E}|x_{i}-x_{j}| \Big] 
  \ge 
  1-\tfrac{1}{s} ,
\]
in time $O(d\epsilon^{-2}s^{3}\log s)$ for $s=O(\frac{d}{\mu_2}\log(\epsilon^{-1}\frac{d}{\mu_2}||b||_{0}))$.
\end{thm}

\paragraph{Remark.} 
If we allow preprocessing of $G$, the runtime of Algorithm~\ref{alg:LaplacianSolver} can be reduced to $O(\epsilon^{-2}s^2)$, as follows. At the preprocessing phase, compute $(\frac{1}{d}A)^t$ for all powers $t\leq s$. Then, instead of approximating $e_u^\tran (\frac{1}{d}A)^tb$ for all powers $t\leq s$, sample a uniform $t\in \set{0,1,...,s}$, and then, in $O(1)$ time (because the probability vector is precomputed, see~\cite{Wal77}),
sample $z\in [n]$ based on the probability vector $e_u^\tran (\frac{1}{d}A)^t$, and finally, output
$\frac{s+1}{d}b_z$. The expectation of the output is $\frac{1}{d}\sum_{t=0}^s (\frac{1}{d}A)^tb$. As for concentration, since the output is in $[-\frac{s+1}{d}\cdot ||b||_\infty,\frac{s+1}{d}\cdot ||b||_\infty]$, by the Hoeffding bound, $O(\epsilon^{-2}s^2)$ many repetitions suffice to obtain (with constant probability) an approximation with additive error $\frac{\epsilon}{2d} ||b||_\infty$
(as in Claim~\ref{claim:Laplacian_power_approximation}).

\medskip
We still need to show how to remove the assumption that $A$ 
has no negative eigenvalues.
Given an adjacency matrix $A$ which might have negative eigenvalues, 
consider the PSD matrix $A'=A+dI$, 
which is the adjacency matrix of the $2d$-regular graph $G'$ 
obtained from $G$ by adding $d$ self-loops to each vertex.
Observe that $A'=U(\Lambda+dI)U^\tran $ and we can write $L=dI-A=(2dI-A')$, 
and thus, similarly to Fact~\ref{fact:inverse_power_series},
$L^+x=\frac{1}{2d}\sum_{t=0}^\infty (\frac{1}{2d}A')^t$,
for $x\in\R^n$ that is orthogonal to the all-ones vector.
Therefore, if we use $A'$ (which is PSD) to guide Algorithm~\ref{alg:LaplacianSolver}'s random walks
(i.e., at each step of a walk, with probability $1/2$ the walk stays put 
and with probability $1/2$ it moves to a uniform neighbor in $G$)
and apply Claims~\ref{claim:Laplacian_tail_bound}
and~\ref{claim:Laplacian_power_approximation} 
(which apply even when $A$ has self-loops), 
an estimate $\hat\delta_{u,v}$ satisfying with high probability
$|\frac{1}{2}\hat\delta_{u,v}-\chi_{u,v}^\tran L^+b| \leq \epsilon\max_{ij\in E} |x_i-x_j|$ is obtained.
When running Algorithm~\ref{alg:LaplacianSolver} on $G'$, the term $s$ evaluates to
$O(\frac{2d}{2d-(\lambda_2+d)}\log(\epsilon^{-1}\frac{2d}{2d-(\lambda_2+d)}||b||_{0}))
=O(\frac{d}{\mu_2}\log(\epsilon^{-1}\frac{d}{\mu_2}||b||_{0}))$,
thus, leaving the guarantee of Theorem~\ref{thm:Laplacian_solver} intact (up to constant factors).

\begin{proof}[Proof of Theorem \ref{thm:Laplacian_solver_intro}]
The theorem follows by a simple modifications to the analysis above.
Observe that the analysis in Claims~\ref{claim:Laplacian_tail_bound} 
and~\ref{claim:Laplacian_power_approximation} holds also when 
replacing $\mu_2$ by a lower bound on $\mu_2$,
which in turn is easy to derive from the upper bound $\barkappa$ 
given in the input and $d$ given as part of input $G$. 
Similarly, $||b||_{0}$ can be replaced by an upper bound $B_{up}\geq||b||_{0}$. 

To handle one vertex $u\in[n]$ instead of two vertices $u,v\in[n]$,
ignore the part dealing with $v$ in Algorithm~\ref{alg:LaplacianSolver},
and modify the analysis in the two aforementioned claims 
to use $e_u$ instead of $\chi_{u,v}$. 
The error bound obtained from combining these lemmas is 
$\frac{\epsilon}{d}||b||_{\infty}$, 
but since each 
$
  \abs{b_i} 
  = \abs{\sum_j L_{ij} x_j} 
  \le \sum_j \abs{L_{ij}}\cdot \norm{x}_\infty
  = 2d \norm{x}_\infty 
$,
we can bound the error by
$\frac{\epsilon}{d}||b||_{\infty} \le 2\epsilon \norm{x}_\infty$. 
\end{proof}

%% file: psdLowerBound.tex
\newcommand{\rtree}{{r_{\mathrm{tree}}}}

\section{Lower Bound for PSD Matrices}
\label{sec:lbPSD}

In this section we prove Theorem~\ref{thm:lbPSD}. 
The entire proof relies on a $d$-regular $n$-vertex graph $G_1$, such that
(i) its girth is $\Omega(\log_d n)$; and 
(ii) its adjacency matrix $A_1$ has 
eigenvalues $\lambda_1\ge \ldots \ge \lambda_n$ that satisfy 
$\maxx{\abs{\lambda_2},\abs{\lambda_n}} \leq \tfrac14 d^{2/3}$
(this bound is somewhat arbitrary, chosen to simplify the exposition). 
We actually need such a graph to exist for infinitely many $n$,
with $d$ bounded uniformly (as $n$ grows).
Such graphs are indeed known, 
for example the Ramanujan graphs constructed 
by Lubotzky, Philips and Sarnak~\cite{LPS88} and by Margulis~\cite{Margulis88} 
for the case where $d-1$ is a prime, 
have eigenvalue upper bound $2\sqrt{d-1}$ 
and girth lower bound $(4/3-o(1))\log_{d-1} n$
(see e.g.~\cite{HLW06}). 

In what follows, let $G_2$ be a certain isomorphic copy of $G_1$
(i.e., obtained from $G_1$ by permuting the vertices, as explained below).
It will be convenient to assume that $G_1$ and $G_2$ 
have the \emph{same} vertex set, which we denote by $V$,
as then we can consider the multi-graph obtained by their edge union,
denoted $G_1\cup G_2$. 
Denoting the adjacency matrix of each $G_i$ by $A_i$,
the adjacency matrix of their edge union $G_1\cup G_2$ is simply $A_1+A_2$. 
We can similarly view $A_1-A_2$ as the adjacency matrix of the same graph,
except that now the edges are signed --- 
those from $G_1$ are positive, and those from $G_2$ are negative. 

The proof of the theorem will follow easily from the three propositions below.
Proposition~\ref{prop:treelike} provides combinatorial, girth-like, 
information about $G_1\cup G_2$.
Proposition~\ref{prop:condition} provides spectral information, 
like the condition number, about $A_1-A_2$.
These two propositions are proved by straightforward arguments,
and the heart of the argument is in Proposition~\ref{prop:queryb}, 
that constructs a PSD linear system based on $A_1-A_2$,
in which the coordinates of the solution $x$ can be analyzed,
showing that recovering a specific coordinate, even approximately,
requires many probes to $b$.

\begin{prop} [Proved in Section~\ref{sec:treelike}]
\label{prop:treelike}
Let $G_1$ be as above and fix a vertex $\hatw\in V$.
Then there exists an isomorphic copy $G_2$ of $G_1$ (on the same vertex set), 
such that in their edge-union $G_1\cup G_2$,
the neighborhood of $\hatw$ of radius $\rtree\eqdef 0.2\log_{4d} n$ 
is a $2d$-regular tree. 
\end{prop}

\begin{prop} \label{prop:condition}
Let $A_1,A_2$ be the adjacency matrices described above, 
and let $\mu \eqdef 2\norm{A_1-A_2}$. 
Then $\mu \leq \tfrac12 d^{2/3}$,
and the matrix $M \eqdef \mu I + A_1-A_2 \in\R^{n\times n}$ is PSD 
with all its eigenvalues in the range $[\tfrac12\mu,\tfrac32\mu]$.
Thus, $M$ is invertible and has condition number $\kappa(M) \leq 3$.
\end{prop}
\begin{proof}
By the triangle inequality, 
$ \mu/2
  =   \norm{A_1-A_2} 
  \le \norm{A_1-dI} + \norm{-(A_2-dI)}
  \le 2\maxx{\abs{\lambda_2},\abs{\lambda_n}} 
  \le \tfrac12 d^{2/3}
$. 
The eigenvalues of $A_1-A_2$ are in the range $[-\tfrac12\mu,\tfrac12\mu]$,
and thus those of $M$ are in the range $[\tfrac12\mu,\tfrac32\mu]$. 
\end{proof}

\begin{prop} [Proved in Section~\ref{sec:queryb}]
\label{prop:queryb}
Let the graphs $G_1,G_2$ be according to Proposition~\ref{prop:treelike},
let $M \eqdef \mu I + A_1-A_2 \in\R^{n\times n}$ as above,
and fix $r\leq \rtree/d^2$. 
Then every randomized algorithm that, with probability at least $6/7$, 
approximates coordinate $x_{\hatw}$ of $x=M^{-1}b$ within additive error 
at most $\tfrac15 \norm{x}_\infty$, %
must probe $d^{\Omega(r)}$ entries from $b\in\R^n$,
even when restricted to $b\in\set{-1,0,+1}^n$ that are supported only on 
vertices at distance $r$ from $\hatw$ (in $G_1\cup G_2$). 
\end{prop}
We can now prove Theorem~\ref{thm:lbPSD} using the above 3 propositions.
Let $G_1$,$G_2$,$A_1$,$A_2$ and $M$ be as required for these propositions,
and fix $r= \rtree/d^2$. 
Let $S\eqdef M$ and observe that it has the sparsity and condition number
required for Theorem~\ref{thm:lbPSD}, 
and let the distinguished index be $u \eqdef \hatw$. 
Now consider a randomized algorithm that, given an input $b\in\R^n$,
estimates coordinate $x^*_u$ of $x^*=S^{-1}b$,
or in other words, coordinate $x_{\hatw}$ of $x=M^{-1}b$.
We can then apply Proposition~\ref{prop:queryb} 
and deduce that this algorithm must probe $b\in\R^n$ in
\[
  d^{\Omega(r)}
  \ge d^{\Omega((\log_{4d} n)/d^2)}
  \ge n^{\Omega(1/d^2)}
\] 
entries, which proves Theorem~\ref{thm:lbPSD}. 

\subsection{Proof of Proposition~\ref{prop:queryb}}
\label{sec:queryb}

Let $V_k\subset V$ be the set of vertices at distance exactly $k$ from $\hatw$
in the edge-union graph $G_1\cup G_2$. 
By the Proposition~\ref{prop:treelike}, 
we can view the radius-$\rtree$ neighborhood of $\hatw$ 
as a tree rooted at $\hatw$.
In particular, for all $k\le \rtree$ 
we have $\card{V_k} = 2d(2d-1)^{k-1} \simeq (2d-1)^k$.
For each vertex $v\in V_k$, 
let $s_v\in\set{\pm1}$ be the value of entry $(\hatw,v)$ in $(A_2-A_1)^k$, 
i.e., the product of the signs along the unique length-$k$ walk 
from $\hatw$ to $v$ in $A_2-A_1$ (i.e., the shortest path in $G_1\cup G_2$). 

Now generate a random $b\in\set{-1,0,+1}^n$ as follows.
First pick an unknown (or random) ``signal'' $\sigma\in\set{\pm1}$;
then use it to choose for each $v\in V_{r}$, 
a random $b_v\in \set{\pm 1}$ with a small bias $\delta>0$ (determined below)
towards $\sigma s_v\in\set{\pm1}$, i.e., 
\[
  \EX[b_v|\sigma] 
  = (\tfrac12 + \tfrac{\delta}{2})\sigma s_v
  + (\tfrac12 - \tfrac{\delta}{2})(-\sigma s_v)
  = \delta\sigma s_v.
\]
Observe that $\EX[s_v b_v | \sigma] = s_v (\delta\sigma s_v) = \delta\sigma$,
which means that $s_v b_v$ has a small bias towards the signal $\sigma$.
Finally, let all other entries be $0$, i.e., $b_v=0$ for $v\notin V_{r}$. 
Observe that $\norm{b}_2^2 = \card{\supp(b)} = \card{V_{r}}$
and $\EX[\sigma \sum_{v\in V_r} s_v b_v \mid \sigma] = \delta \card{V_r}$. 
We set the bias to be
\begin{equation} \label{eq:bias} 
  \delta \eqdef C (r^2\log d)\ \card{V_{r}}^{-1/3} 
\end{equation}
for a sufficiently large constant $C>0$.
Notice that $\set{s_v}_{v\in V_r}$ have fixed values known to the algorithm,
hence observing $b_v$ (by probing this entry of $b$) 
is information-theoretically equivalent to observing $s_v b_v$. 

The next lemma is standard and follows easily from Yao's minimax principle,
together with a bound on the total-variation distance between two Binomial
distributions, with biases $\tfrac12+\delta$ and $\tfrac12-\delta$),
see e.g. \cite[Fact D.1.3]{Cannone15} or \cite[Eqn.~(2.15)]{Adell2006}.

\begin{lemma} \label{lem:BinomProbes}
Every randomized algorithm that, with probability at least $1/2+\gamma$
for $\gamma\in (0,1/2)$, recovers an unknown signal $\sigma\in\set{\pm 1}$ 
from $b_1,b_2,\ldots \in\set{\pm 1}$,
each set independently to $\sigma$ or $-\sigma$ with bias $\delta>0$, 
must probe at least $\Omega(\delta^{-2}\gamma^2)$ entries of $b$. 
\end{lemma}

We proceed to analyze $x_{\hatw}$, 
aiming to show that it can be used to recover $\sigma$, 
namely, that with high probability $\sgn(x_{\hatw})=\sigma$. 
Later we will bound $\norm{x}_\infty$ 
aiming to show a similar conclusion for $x_{\hatw} \pm \tfrac15 \norm{x}_\infty$. 
For convenience, denote $B \eqdef \frac{A_2-A_1}{\mu}$,
hence $\norm{B} = \tfrac{\mu/2}{\mu}=\tfrac12$ and 
\[
  M^{-1} 
  = (\mu (I-\tfrac{A_2-A_1}{\mu}))^{-1}
  = \mu^{-1} \sum_{i\ge0} B^i,
\]
and since $B$ is symmetric, for every vertex $u\in V$ (including $\hatw$), 
\begin{equation} \label{eq:xu}
  x_u
  = \tuple{e_u, M^{-1}b}
  = \mu^{-1} \sum_{i\ge 0} \tuple{e_u,B^i b}
  = \mu^{-1} \sum_{i\ge 0} b^\tran B^i e_u . 
\end{equation}
Each summand $b^\tran B^i e_u$ can be viewed as the summation, 
over all length-$i$ walks from vertex $u$,
of the coordinate $b_v$ corresponding to the walk's end-vertex $v$, 
multiplied by $\mu^{-i}$ 
and by the product of the signs of $A_2-A_1$ along the walk. 
We can restrict the summation to walks ending at vertices $v\in V_{r}$, 
as otherwise $b_v=0$. 

\begin{lemma} \label{lem:ilarge}
For every vertex $u\in V$ (including $\hatw$), 
\[
  \sum_{i\ge 2r\log\mu} \Big| b^\tran B^i e_u \Big|
  \le 
  \tfrac14 \mu^{-2r}  \cdot\delta\card{V_{r}} .
\]
\end{lemma}
\begin{proof}[Proof of Lemma~\ref{lem:ilarge}]
For each $i$, we have by Cauchy-Schwartz 
$
  \abs{b^\tran B^i e_u}
  \le \norm{b}_2 \cdot \norm{B}_2^i
  \le \card{V_{r}}^{1/2}\cdot 2^{-i} 
$, 
and then by our choice of the bias $\delta$ in \eqref{eq:bias},  
\[
  \sum_{i\ge 2r\log\mu} \abs{b^\tran B^i e_u}
  \le \card{V_{r}}^{1/2} \sum_{i\ge 2r\log\mu} 2^{-i} 
  \le (\card{V_{r}} \cdot \delta/8) \cdot 2\mu^{-2r} .
\]
\end{proof}

Recall that by Proposition~\ref{prop:treelike}, 
the neighborhood of $\hatw$ of radius $\rtree$ is a tree,
and view it as a tree rooted at $\hatw$. 
For a vertex $u$ in this tree, let $S_u$ 
be the set of all vertices $v\in V_r$ that are descendants of $u$;
for example, $S_{\hatw} = V_r$, 
and if the distance of $u$ from $\hatw$ is greater than $r$ then $S_u=\emptyset$.
Define a random variable $Z_u \eqdef \sum_{v\in S_u} s_v b_v$,
whose expectation is 
\[
  \EX[Z_u] 
  = \sum_{v\in S_u} \EX[s_v b_v \mid \sigma] 
  = \card{S_u}\cdot \delta\sigma.
\]

\begin{lemma} \label{lem:ismall}
With probability at least $6/7$, 
\begin{equation} \label{eq:ismalldeviation}
  \forall 0\le k\le r,\ 
  \forall u\in V_k, 
  \qquad
  \big| Z_u -\EX[Z_u] \big| 
  \leq O\left( \sqrt{\card{S_u}} \cdot \ln (3\card{V_k}) \right) .
\end{equation}
\end{lemma}
We remark that the constant $3$ is somewhat arbitrary but needed to make sure 
the righthand-side is positive even for $k=0$ (as $\card{V_0}=1$). 
In addition, applying \eqref{eq:ismalldeviation} to $\hatw\in V_0$ yields, 
by our choice of the bias $\delta$ in \eqref{eq:bias},  
\begin{equation} \label{eq:ismalldeviation_w}
  \big| Z_{\hatw} -\EX[Z_{\hatw}] \big| 
  \le O\left( \sqrt{\card{V_r}} \cdot \ln (3\card{V_r}) \right) 
  \le \tfrac14 \delta \card{V_r}.
\end{equation}

\begin{proof}[Proof of Lemma~\ref{lem:ismall}]
Fix $0\le k\le r$ and $u\in V_k$. 
By Hoeffding's inequality, for every $c>0$, 
\[
  \Pr\Big[ \abs{Z_u - \EX[Z_u]} \ge c \sqrt{\card{S_u}\ln (3\card{V_k})} \Big] 
  \le e^{- 2 c^2 \card{S_u}\ln (3\card{V_k}) / (4\card{S_u}) }
  \le e^{- (c^2/2) \ln (3\card{V_k}) } 
  = (3\card{V_k})^{- c^2/2}. 
\]
By a union bound over all $u\in V_0\cup\cdots\cup V_r$, 
\[
  \Pr\Big[ \exists u,\ \abs{Z_u - \EX[Z_u]} 
              \ge c \sqrt{\card{S_u}\ln (3\card{V_k})} 
     \Big] 
  \le \sum_{k=0}^r \card{V_k} \cdot (3\card{V_k})^{- c^2/2 }
  = \tfrac13 \sum_{k=0}^r (3\card{V_k})^{1 - c^2/2 }. 
\]
For all $c\ge 2$ this series is decreasing geometrically, 
because $\card{V_k}$ grows at least by a factor of $2d-1 \ge 5$,
and thus the sum is dominated by its first term. 
By choosing $c$ to be an appropriate constant,
the first term (and the entire sum) can be made arbitrarily small.
\end{proof}

We assume henceforth that the event described in Lemma~\ref{lem:ismall} occurs.
Let $W_i$ be the set of all walks of length $i$ that start at $\hatw$ 
and end (at some vertex) in $V_r$, 
i.e., at distance exactly $r$ from $\hatw$.
Define 
\[
  Q\eqdef \sum_{i=r}^{5r\log \mu} \mu^{-i}\card{W_i}. 
\]
We make two remarks.
First, we can equivalently start the summation from $i=0$,
because $W_i=\emptyset$ for all $i<r$.
Second, the leading constant $5$ here is bigger than 
the $2$ used in Lemma~\ref{lem:ilarge},
this is intentional and the slack be needed at the very end of the proof.

\begin{lemma} \label{lem:xw}
If the event in Lemma~\ref{lem:ismall} occurs, then 
\[
  x_{\hatw} \in (\sigma \pm\tfrac12) \delta\cdot \mu^{-1}Q, 
\]
and thus $\sgn(x_{\hatw})=\sigma$ (i.e., recovers the signal). 
\end{lemma}

\begin{proof}[Proof of Lemma~\ref{lem:xw}]
We would like to employ~\eqref{eq:xu} 
and the interpretation of $b^\tran B^i e_{\hatw}$ via walks of length $i$. 
To this end, fix $0\le i\le 5r\log \mu$.
Observe that $i\le \rtree$, 
hence a walk of length $i$ from $\hatw$ is entirely contained in 
the $2d$-regular tree formed by the neighborhood of $\hatw$ of radius $\rtree$.
Each such walk contributes the value $b_v$ at the walk's end vertex $v$,
multiplied by all the signs seen along the walk.
We make two observations.
First, we can restrict attention to end vertices $v\in V_r$
(and in particular $i\ge r$), because otherwise $b_v=0$.
Moreover, the same number of walks end at each $v\in V_r$, by symmetry. 
Second, the signs along a walk in a tree cancel, 
except for the signs on the shortest path between $\hatw$ and $v$
(the start and end vertices), 
hence the product of these signs is exactly $s_v$.
By symmetry, the number of walks ending at each $v\in V_r$ is the same, 
namely, $\tfrac{\card{W_i}}{\card{V_r}}$, 
and thus 
\begin{equation} \label{eq:xu_hatw}
  b^\tran B^i e_{\hatw}
  = \sum_{v\in V_r} \tfrac{\card{W_i}}{\card{V_r}} \mu^{-i} s_v b_v 
  = \tfrac{Z_{\hatw}}{\card{V_r}} \cdot \mu^{-i} \card{W_i} . 
\end{equation}
Assuming the event in Lemma~\ref{lem:ismall} occurs, 
we have 
$
  Z_{\hatw} 
  \in (\delta\sigma|S_{\hatw} | \pm \tfrac14 \delta \card{V_r}) 
  =   (1\pm \tfrac14) \sigma\delta \card{V_r} 
$,
and therefore (recall terms for $i<r$ have zero contribution)
\[
  \sum_{i=0}^{5r\log\mu} b^\tran B^i e_{\hatw}
  \in \sum_{i=r}^{5r\log\mu} (1\pm \tfrac14) \sigma\delta\cdot \mu^{-i} \card{W_i} 
  = (1\pm \tfrac14) \sigma\delta\cdot Q.
\]

For the range of $i>5r\log\mu$, we can use Lemma~\ref{lem:ilarge} 
and the obvious $\card{W_r} = \card{V_r}$ to derive 
\[
  \Big| \sum_{i > 5r\log\mu} b^\tran B^i e_{\hatw} \Big|
  \le \sum_{i > 5r\log\mu} \Big| b^\tran B^i e_{\hatw} \Big|
  \le \tfrac14 \mu^{-2r}  \cdot\delta\card{V_{r}} 
  \le \tfrac14 \delta Q.
\]
Altogether, plugging into~\eqref{eq:xu} we obtain 
\[
  \mu\cdot x_{\hatw} 
  = \sum_{i\ge 0} b^\tran B^i e_{\hatw} 
  \in \sum_{i=0}^{5r\log\mu} (1\pm\tfrac14) \sigma\delta \cdot Q
      \pm \tfrac14 \delta\cdot Q
  = (1\pm\tfrac12) \sigma\delta\cdot Q,  
\]
which proves the lemma because $\sigma\in\set{\pm1}$.
\end{proof}

\begin{lemma} \label{lem:xinfty}
If the event in Lemma~\ref{lem:ismall} occurs, then 
\[
  \norm{x}_\infty \le 2\delta\cdot \mu^{-1}Q.
\]
\end{lemma}

\begin{proof}
Fix $u\in V$, and let us bound $\abs{x_{u}}$. 
Similarly to the proof of Lemma~\ref{lem:xw}, 
we employ~\eqref{eq:xu} and interpret $b^\tran B^i e_{u}$ via walks of length $i$,
which now start at vertex $u$ rather than at $\hatw$. 

Let $k$ be the distance of $u$ from $\hatw$, i.e., $u\in V_k$.
The case $k > r + 2r\log\mu$ is easy, as follows.
A walk of length $i\le 2r\log\mu$ from $u\in V_k$ cannot end in $V_r$ 
(because the distance of the end vertex from $\hatw$ is at least $k-i > r$),
hence $b^\tran B^i e_{u}=0$.
Plugging this information and Lemma~\ref{lem:ilarge} into~\eqref{eq:xu}, 
we have
\[
  \abs{ \mu\cdot x_{u} }
  \le \sum_{i \ge 0} \Big| b^\tran B^i e_{u} \Big|
  =   \sum_{i > 2r\log\mu} \Big| b^\tran B^i e_{u} \Big|
  \le \tfrac14 \mu^{-2r}  \cdot\delta\card{V_{r}} 
  \le \tfrac14 \delta Q,
\]
which proves the lemma in this case. %

We thus assume henceforth $k \le r + 2r\log\mu$.
For each $i \le 2r\log\mu$, let $U_i$ be the set of all length-$i$ walks 
that start at $u$ and end in $V_r$.
(The difference from $W_i$ is that the walks start at $u$ instead of $\hatw$.) 
Observe that such a walk (in $U_i$) is entirely contained in 
the $2d$-regular tree formed by the radius-$\rtree$ neighborhood of $\hatw$,
because the maximum distance from $\hatw$ it can reach is 
$ 
  k+i 
  \le (r+2r\log\mu) + 2r\log\mu 
  \le 5r\log d 
  \le \rtree
$. 
In addition, we claim that 
\begin{equation} \label{eq:Ui}
  \card{W_{i+k}} 
  \ge \card{V_k}\cdot \card{U_i} 
  \ge (2d-1)^k \card{U_i} .
\end{equation}
Indeed, we can generate walks in $W_{i+k}$ 
by first walking from $\hatw$ to any vertex in $u'\in V_k$ directly, 
i.e., along the unique shortest path, 
and then ``imitating'' a walk from $U_i$,
in the sense of executing it from $u'\in V_k$ instead of from $u\in V_k$. 
(Formally, view a walk in $U_i$ as a sequence in $[2d]^i$ 
that determines which outgoing edge to traverse next, 
according to a fixed numbering of the $2d$ incident edge, 
that reserves $2d$ to the edge that gets us closer to $\hatw$, if one exists.) 
This yields $\card{V_k}\cdot \card{U_i}$ walks that are all distinct 
and end in $V_r$, hence these are distinct walks in $W_{i+k}$.

Denote the shortest path from $\hatw$ to $u$ by $u_0=\hatw,u_1,\ldots,u_k=u$;
notice that these vertices are exactly the ancestors of $u$ 
when we view the neighborhood of $\hatw$ as a tree rooted at $\hatw$.
Now partition $U_i = U_{i,0}\cup\cdots\cup U_{i,k}$ by letting each $U_{i,j}$ 
contain the walks in $U_i$ that visit $u_j$ but not $u_0,\ldots,u_{j-1}$,
which means that $u_j$ is the farthest (from $u$) ancestor visited by the walk.
For example, $U_{i,0}$ contains all walks in $U_i$ that visit $u_0=\hatw$,
and $W_{i,k}$ contains all walks in $U_i$ that never visit $u_{k-1}$
and thus never visit $V_0\cup\cdots\cup V_{k-1}$
(of course, for $k>r$ this cannot happen and thus $U_{i,k}=\emptyset$). 
The walks in $U_{i,j}$ all end in $S_{u_j}$ 
(recall this is the set of vertices in $V_r$ that are descendants of $u_j$),
and by symmetry, the number of walks ending at each $v\in S_{u_j}$ is the same, 
namely, $\frac{\card{U_{i,j}}}{\card{S_{u_j}}}$, 
and thus similarly to~\eqref{eq:xu_hatw}, 
\[
  b^\tran B^i e_{u}
  = \sum_{j=0}^k \sum_{v\in S_{u_j}} \frac{\card{U_{i,j}}}{\card{S_{u_j}}} \mu^{-i} s_v b_v 
  = \sum_{j=0}^k \frac{Z_{u_j}}{\card{S_{u_j}}} \mu^{-i} \card{U_{i,j}} .
\]
Assuming the event in Lemma~\ref{lem:ismall} occurs, 
\[
  \forall j=0,\ldots,k, 
  \qquad
  \frac{Z_{u_j}}{\card{S_{u_j}}}
  \in \delta\sigma 
      \pm O\Big(\frac{ \ln(3\card{V_k}) }{ {\card{S_{u_j}}^{1/2}} } \Big) ,
\]
and thus 
\[
  \Big| b^\tran B^i e_{u} \big| 
  \le \sum_{j=0}^k \Big| \frac{Z_{u_j}}{\card{S_{u_j}}} \mu^{-i} \card{U_{i,j}} \Big| 
  \le \sum_{j=0}^k \Big(\delta + \frac{\ln(3\card{V_k})} {{\card{S_{u_j}}^{1/2}}} \Big) 
                  \mu^{-i} \card{U_{i,j}} .
\]
Recall that $u_j$ is an ancestor of $u$, hence 
$
  \card{S_{u_j}} 
  \geq \card{S_u} 
$,
and use~\eqref{eq:Ui} to obtain 
$\sum_{j=0}^k \card{U_{i,j}} = \card{U_i} \leq \frac{1}{(2d-1)^k} \card{W_{i+k}}$,
and altogether we have
\[
  \Big| b^\tran B^i e_{u} \big| 
  \le \Big(\delta + \frac{\ln(3\card{V_k})} {{\card{S_u}^{1/2}}} \Big) 
      \mu^{-i} \sum_{j=0}^k \card{U_{i,j}} 
  \le \Big(\delta + \frac{\ln(3\card{V_k})} {{\card{S_u}^{1/2}}} \Big) 
      \Big( \frac{\mu}{2d-1} \Big)^k \mu^{-(i+k)} \card{W_{i+k}} .
\]
To simplify notation, define the quantity (notice it does not depend on $i$) 
\begin{equation} \label{eq:alpha}
  \alpha_k
  \eqdef \Big(\delta + \frac{\ln(3\card{V_k})} {{\card{S_u}^{1/2}}} \Big) 
         \Big( \frac{\mu}{2d-1} \Big)^k.
\end{equation}
We claim that $\alpha_k \le \tfrac32\delta$. 
To prove this claim, we first easily bound one part
$\delta (\frac{\mu}{2d-1})^k \leq \delta$. 
For the other part, observe that 
$
 \card{V_r}
 = \card{S_{\hatw}} 
 \le 4(2d-1)^k \card{S_u} 
$,
and thus 
\[
  \frac{\ln(3\card{V_k})} {{\card{S_u}^{1/2}}} \Big( \frac{\mu}{2d-1} \Big)^k
  \le \frac{2\ln(3\card{V_k})} {{\card{V_r}^{1/2}}} 
      \Big( \frac{\mu^2}{2d-1} \Big)^{k/2}
  \le \frac{2\ln(3\card{V_r})} {{\card{V_r}^{1/2}}} 
      \Big( \frac{\mu^2}{2d-1} \Big)^{r/2}
  \le \frac12 \delta ,
\]
where the last inequality is by our assumption 
$\mu \leq \tfrac12 d^{2/3}$, 
which implies
$
  (\frac{\mu^2}{2d-1})^{r/2}
  \le (2d-1)^{1/3\cdot r/2}
  \le \card{V_r}^{1/6}
$,
and by our choice of the bias $\delta$ in \eqref{eq:bias}. 
Putting these bounds together proves the claim.

With this bound $\alpha_k \le \tfrac32\delta$ in hand, 
we are finally ready to conclude the lemma. 
Using this claim and that $k \le r+2r\log\mu \le 3r\log\mu$, 
\[
  \sum_{i \le 2r\log\mu} \Big| b^\tran B^i e_{u} \Big|
  \le \sum_{i \le 2r\log\mu} \alpha_k \mu^{-(i+k)} \card{W_{i+k}} 
  \le \alpha_k Q 
  \le \tfrac32 \delta Q. 
\]
For the range of $i>2r\log\mu$, we can use Lemma~\ref{lem:ilarge} 
and the obvious $\card{W_r} = \card{V_r}$ to derive 
\[
  \sum_{i > 2r\log\mu} \Big| b^\tran B^i e_u \Big|
  \le \tfrac14 \mu^{-2r}  \cdot\delta\card{V_{r}} 
  \le \tfrac14 \delta Q.
\]
Plugging the above information into~\eqref{eq:xu},
we have
\[
  \abs{ \mu\cdot x_{u} }
  \le \sum_{i \ge 0} \Big| b^\tran B^i e_{u} \Big|
  \le 2 \delta Q,
\]
which concludes the case $k \le r + 2r\log\mu$,
and completes the proof of Lemma~\ref{lem:xinfty}. 
\end{proof}

We can now complete the proof of Proposition~\ref{prop:queryb}.
By Lemma~\ref{lem:ismall}, with probability at least $6/7$
the event described therein occurs.
Assume this is the case and consider an estimate $\hat x_{\hatw}$ for $x_{\hatw}$
that has additive error at most $\eps \norm{x}_\infty$ for $\eps\le \tfrac15$. 
By Lemma~\ref{lem:xw} we have 
$
  x_{\hatw} \in (\sigma \pm\tfrac12) \delta\cdot \mu^{-1}Q
$, 
and by Lemma~\ref{lem:xinfty} we have 
$
  \norm{x}_\infty \le 2\delta\cdot \mu^{-1}Q
$.
Altogether 
\[
  \hat x_{\hatw}   
  \in x_{\hatw} \pm \tfrac15 \norm{x}_\infty
  \subseteq 
  (\sigma \pm \tfrac12 \pm \tfrac25) \delta\cdot \mu^{-1}Q ,
\]
which implies that $\sgn(x_{\hatw})=\sigma$. 

Now consider a randomized algorithm for estimating $x_{\hatw}$, 
and whose output $\hat x_{\hatw}$ satisfies the above additive bound 
with probability at least $6/7$. 
We can use this estimation algorithm to recover the signal $\sigma$, 
by simply reporting the sign of its estimate, namely $\sgn(x_{\hatw})$.
This recovery does not require additional probes to $b$, 
and by a union bound, 
it succeeds (in recovering $\sigma)$ with probability at least $5/7$. 
But by Lemma~\ref{lem:BinomProbes}, such a recovery algorithm,
and in particular the algorithm for estimating $x_{\hatw}$, 
must probe $b$ in at least
\[
  \Omega(\delta^{-2}) 
  \ge \Omega\left( \card{V_{r}}^{2/3}/(r^4\log^2 d) \right)
  \ge \Omega\left( (2d-1)^{2r/3} / (r^4\log^2 d) \right)
  \ge d^{\Omega(r)} 
\]
entries, which proves Proposition~\ref{prop:queryb}.

\subsection{Proof of Proposition~\ref{prop:treelike}}
\label{sec:treelike}

We prove Proposition~\ref{prop:treelike} by the probabilistic method,
namely, we let $G_2$ be a random isomorphic copy of $G_1$,
and argue that the desired property of $G_1\cup G_2$ 
holds with positive, in fact high, probability. 
The edge-union graph $G_1\cup G_2$ can be equivalently constructed as follows. 
Start with the fixed graph $G_1$ as above and a copy of it $G'$
on a disjoint set of vertices, i.e., $V(G_1)\cap V(G')=\emptyset$;
now draw a random perfect matching $M$ between $V(G_1)$ and $V(G')$,
and then contract every edge of $M$.
The rest of the proof considers the graph \emph{prior} to the contraction. 
The idea is to expose the edges of $M$ gradually,
and then by the principle of deferred decision, 
the remaining edges of $M$ form a random perfect matching 
between the yet unmatched vertices in $G$ and $G'$. 

A walk in the edge-union graph $G_1\cup G_2$
can be viewed as a walk in the prior-to-contraction graph, 
except that moves along the matching $M$ are not counted as steps. 
Fixing a vertex $\hatw\in V$ and an integer $l\ge 1$, 
every length-$l$ walk starting at $\hatw$ 
can be associated with a distinct sequence $a\in [1..2d]^l$ as follows. 
For each step $i=1\,\ldots,l$, move from the current vertex 
(starting at $\hatw\in V$) along an edge represented by $a_i$,
where $a_i\in [1..d]$ corresponds to an edge in $G$,
and $a_i\in[d+1..2d]$ corresponds to an edge in $G'$. 
To make it more precise, recall that $G$ is fixed,  
hence the $d$ edges of $G$ incident to each vertex in $v\in V(G)$
have a fixed ordering (say lexicographic), 
and can be associated with a distinct index from $[1..d]$. 
(An edge $(u,v)$ may have different indices at $u$ and at $v$.)
The same applies also to $G'$, 
except that now the indices are from $[d+1..2d]$. 
Thus, each index $a_i\in [1..2d]$ represents a transition 
along an edge in $E(G)\cup E(G')$, 
but if the current vertex is not in the ``correct'' graph, 
then the walk first moves along the matching $M$ 
(to cross between $V(G)$ and $V(G')$), 
and only then moves along the edge of $E(G)\cup E(G')$. 
Observe that step $i$ in the walk does not exposes any edge of $M$
if both $a_{i-1},a_i\in[1..d]$ or both $a_{i-1},a_i\in[d+1..2d]$, 
where by convention $a_{0}\in[1..d]$. 

For the neighborhood of $\hatw$ of radius $\rtree$ in $G_1\cup G_2$ 
to \emph{not} be a $2d$-regular tree,
obviously there must exist $a\in [1..2d]^{2\rtree}$, 
whose corresponding walk is non-backtracking
(i.e., no step $i\ge 2$ moves along the same edge as step $i-1$)
yet it is self-intersecting
(i.e., at least one vertex is visited more than once).
Our analysis employs another necessary condition, 
whose existence follows by identifying a cycle near $\hatw$ in $G_1\cup G_2$, 
and two successive edges in it that originate from different graphs 
(they must exists because each of $G_1$ and $G_2$ has high girth). 
Specifically, there must exist two sequences (walks)
$a\in[1..2d]^{l'}$ and $b\in[1..2d]^{l''}$ of lengths $1\le l',l''\le 2\rtree$, 
such that 
(i) $a_{l'}\in[1..d]$, i.e., the last step of $a$ is in $G$;
(ii) $b_{l''}\in[d+1..2d]$, i.e., the last step of $b$ is in $G'$;
(iii) the last vertex in the walk $a$ and that in $b$ 
are matched to each other by $M$; and 
(iv) these two last vertices were not visited by or matched to
any earlier vertex in the two walks.
Let $E_{a,b}$ denote the event that requirements (i)-(iv) are satisfied.
Observe that $a$ and $b$ may have a common prefix, 
during which they will obviously visit the same vertices.
For instance, a common prefix of length $l'-1=l''-1$ corresponds 
to having two parallel edges (originating from $G_1$ and $G_2$).

We now claim that $\Pr[E_{a,b}] \le \frac{2}{n}$ for every fixed $a,b$
of length $l',l''\le 2\rtree$.
To see this, follow the walks corresponding to $a$ and to $b$, 
and expose the edges of $M$ incident to all visited  vertices 
except for the last vertex in each walk. 
If requirement (iv) is already violated, 
then the probability of $E_{a,b}$ is $0$.
We may thus assume henceforth it is satisfied, 
which implies that after exposing at most $l'+l''-1 \le 4\rtree$ edges of $M$,
the last vertex in each walk is still not matched by $M$.
If we now expose the edges of $M$ incident to the last vertex in each walk,
the probability that these two vertices are matched to each other 
is at most $\frac{1}{n-4\rtree} \le \frac{2}{n}$,
and the claimed bound follows. 

Finally, by a union bound over all possible sequence pairs $(a,b)$,
the probability that the neighborhood of $\hatw$ %
is \emph{not} a $2d$-regular tree, is at most
\[
  \Pr[\vee_{a,b} E_{a,b}]
  \le (2\rtree)^2 (2d)^{2\rtree} \cdot \tfrac{2}{n} 
  \le (4d)^{2\rtree} \cdot \tfrac{8}{n} 
  \le \tfrac{1}{\sqrt n} .
\]
This completes the proof of Proposition~\ref{prop:treelike}. 

%% file: lowerBoundK2.tex
\section{Square of Condition Number is Necessary}
\label{sec:lbK2}

In this section, we prove Theorem~\ref{thm:lbK2}. In particular, we
show that there exist graphs for which one needs to query $b$ at least
$t$ times where $t$ is nearly-quadratic in the condition number of the
Laplacian $L=L_G$.

We first describe the construction of the graph $G$.  Let $X$ be a
3-regular expander on $n/2$ nodes, indexed by a set $V_X$, with girth
$\Omega(\log n)$. We build a graph $G$ as follows. Take two copies of
$X$, termed $X_1$ on vertices $1,\ldots n/2$ and $X_2$ on vertices
$n/2+1,\ldots n$. Then we pick $n/k$ nodes in $X_1$, termed $C_1$, and
the equivalent $n/k$ nodes in $C_2$ (i.e., originating from the same
nodes in $V_X$), and connect $C_1$ to $C_2$ via a matching $M$. Let
$L$ be the Laplacian of the resulting graph $G$.

\begin{lemma}
The condition number of $L$ is $O(k)$.
\end{lemma}

\begin{proof}
We need to prove that, for any unit-norm $x$ orthogonal to all 1s, we
have that $x^\tran Lx\ge \Omega(1/k)$ as the largest eigenvalue is
$\Theta(1)$. We can decompose the Laplacian $L$ into 3 components,
corresponding to $X_1, X_2$ and $M$: $L=L_1+L_2+L_M$. Similarly,
decompose $x=x'+x''$ corresponding to vertices $V_1$ and $V_2$. Let
$m=\tfrac{n}{2}\sum_{i\in[n/2]} x'_i=-\tfrac{n}{2}\sum_{i>n/2} x''_i$,
and $\bar x'$ be $x_1$ minus $m$ on the $V_1$ coordinates. Similarly
$\bar x''$ is $x_2$ plus $m$ on the $V_2$ coordinates. Then, we have
that:
$$ x^\tran Lx=x'^\tran L_1x'+(x'')^\tran L_2x''+x^\tran L_Mx.
$$

For the sake of contradiction, suppose that all three terms are
$<c/k$ for small $c>0$. Then, using the fact $X_1$ is expander,
$\tfrac{1}{2}\|\bar x'\|^2\le (x')^\tran L_1x'\le c/k$, i.e., $\|\bar
x'\|\le \sqrt{2c/k}$. Similarly $\|\bar x''\|\le \sqrt{2c/k}$. Also we have
that:
$$
x^\tran Mx=\sum_{(i,j)\in M} (x_i-x_j)^2=\sum_{(i,j)\in M} (2m+\bar
x'_i-\bar x''_j)^2\ge (2m\cdot \sqrt{n/k}-\|\bar x'\|-\|\bar x''\|)^2,
$$
where the last inequality uses triangle inequality.

Since $\|x\|=1$, we also have that $1=\|x\|\le m\cdot
\sqrt{n}+\|\bar x'\|+\|\bar x''\|$, and thus $m\ge
0.5/\sqrt{n}$. Plugging $m$ into the above, we obtain that $x^\tran Mx\ge
(1/\sqrt{k}-2\cdot \sqrt{2c/k})^2\ge 0.25/k>c/k$ --- a
contradiction.
\end{proof}

To prove the lower bound on the number of probes into $b$, we consider
two distributions on $b$ that can be distinguished using an estimate
to $|x_u-x_v|$. Distinguishing these two distributions will require a
large number of probes into $b$, giving us a query lower bound for
estimating $|x_u-x_v|$.

We now describe these two distributions. Partition the graph $G$ into
4 parts, termed $P_1,\ldots, P_4$ as follows: the vertices $V_1$ are
split arbitrarily into 2 equal-size sets $V_1=P_1\cup P_2$, and similarly with
$V_2=P_3\cup P_4$.  Consider distinguishing the following two cases,
where $p=\Theta(\sqrt{\log n}/k)$. 
\begin{description}
\item[Balanced:]
For each coordinate $u\in P_1, P_3$ pick $b_u\in\set{\pm 1}$ randomly; for
$u\in P_2$ and $u\in P_4$ pick $b_u\in\set{\pm 1}$ randomly conditioned on the
fact that $|\{u\in P_2 \mid b_u=+1\}|=|\{u\in P_1 \mid b_u=-1\}|$
(i.e., $b$ is fully balanced on $V_1$), and similarly for $V_2$. 
\item[Unbalanced:]
For each $u\in P_1$, set $b_u$ to $+1$ with probability $1/2+p/2$ and to $-1$
with probability $1/2-p/2$; for $u\in P_2$, we similarly set $b_u\in\set{\pm
1}$ randomly conditioned on $|\{u\in P_2 \mid b_u=+1\}|=|\{u\in P_1
\mid b_u=+1\}|$ (i.e., the bias is exactly the same in $P_1$ and
$P_2$). For $V_2$, we do the same but with bias $\Pr[b_u=+1]=1/2-p/2$.
\end{description}

Note that distinguishing the two cases with probability $\ge 1/2+\delta$
requires probing $\Omega(\delta^2\cdot 1/p^2)=\Omega(k^2/\log^3 n)$
coordinates of $b$ (see Lemma~\ref{lem:BinomProbes}). Now let us show how 
to distinguish the two cases by computing $x_u-x_v$ for some fixed edge
$(u,v)\in G$, where $x$ is the solution to $Lx=b$.

\begin{lemma}
\label{lem:maxDiff}
In the Balanced case, for any two vertices $u,v$ of $G$, we have that
$x_u-x_v=O(\sqrt{\log n})$ with probability  at least $1-1/n$.

In the Unbalanced case, for any two vertices $u,v$ of $G$, we have that
$x_u-x_v=O(pk\log n)+O(\sqrt{\log n})$ with probability at least $1-1/n$.
\end{lemma}

The lemma is the core of the argument and its proof is deferred to Section~\ref{sec:maxDiff}.
\begin{lemma}
\label{lem:randomDiff}
In the Unbalanced case, if we pick an edge $(u,v)\in M$ at random,
then $|x_u-x_v|=\Omega(\sqrt{\log n})$ with probability at least $\Omega(1/\log
n)$.
\end{lemma}

\begin{proof}
Since $\sum_{u\in V_1} b_u=-\sum_{u\in V_2} b_u=\Omega(np)$ with high
probability, we have that $\tfrac{1}{|M|}\sum_{(u,v)\in M}
|x_u-x_v|\ge \tfrac{1}{|M|}\Omega(np)=\Omega(kp)$. Since
$\max_{u,v}|x_u-x_v|\le O(pk\log n+\sqrt{\log n})=O(pk\log n)$. Hence,
if we pick a random $(u,v)\in M$, we have a probabity of at least
$\Omega(1/\log n)$ that $|x_u-x_v|\ge \Omega(kp)$.
\end{proof}

\begin{proof}[Proof of Theorem~\ref{thm:lbK2}, assuming
    Lemma~\ref{lem:maxDiff}]
The proof of the theorem follows immediately from the above two
lemmas. In particular, for a random edge $(u,v)\in M$, in the balanced
case, we have $|x_u-x_v|\le O(\sqrt{\log n})$ with $1-1/n$
probability. On the other hand, in the unbalanced case, we have
$|x_u-x_v|\ge \Omega(\sqrt{\log n})=\Omega(1/\log n)\cdot \max_{u',v'}
|x_{u'}-x_{v'}|\ge \tfrac{\eps}{2}\cdot \max_{u',v'}
|x_{u'}-x_{v'}|$. Hence we can distinguish the two distributions with
probability at least $\Omega(1/\log n)$ as follows. Suppose the
implicit constants from Lemma~\ref{lem:randomDiff} are respectively
$q,w>0$. Then we estimate $|x_u-x_v|$ and if it's $< q\sqrt{\log n}$,
then output ``balanced'' with probability $1/2+\tfrac{w}{2}/\log
n$. If $|x_u-x_v|\ge q\sqrt{\log n}$, then output ``unbalanced''. This
algorithm has probability of correctly distinguishing balanced vs
unbalanced with probability at least $1/2+\tfrac{w}{2}/\log n$ (as
long as $|x_u-x_v|$ is estimated correctly, which happens with
probability at least $1-O(1/\log n)$ after a standard amflication by
taking the median of $O(\log\log n)$ independent runs of the
algorithm). On the other hand, any such algorithm must make at least
$\Omega((1/\log n)^2\cdot 1/p^2)=\Omega(k^2/\log^3 n)$ coordinates of
$b$ (see Lemma~\ref{lem:BinomProbes}).
\end{proof}

\subsection{Proof of Lemma~\ref{lem:maxDiff}}
\label{sec:maxDiff}

\begin{proof}[Proof of Lemma~\ref{lem:maxDiff}]
Consider the matrix $A$ to be the adjacency matrix of the graph $G$,
where each node $i$ not in the matching $M$ has a self-loop. Thus all
nodes have degree $d=4$. Then we have that $L=dI-A$, and hence
$b=Lx=dIx-Ax$, or
$x=\tfrac{1}{d}b+\tfrac{1}{d}Ax$. Furthermore,
using this identity iteratively, we have:
$$
x=\tfrac{1}{d}b+\tfrac{1}{d}Ax=\tfrac{1}{d}b+\tfrac{1}{d}A(\tfrac{1}{d}b+\tfrac{1}{d}Ax)=\ldots=\sum_{i=0}^{t}
(\tfrac{1}{d}A)^i\cdot (\tfrac{1}{d}b)+(\tfrac{1}{d}A)^\tran x.
$$

Let $x'=\sum_{i=0}^{t} (\tfrac{1}{d}A)^i\cdot (\tfrac{1}{d}b)$.  We
take the solution $x$ such that $x\perp {\allones}$ and hence
$\|x-x'\|\le (1-O(1/k))^t\cdot O(k)\cdot \|b\|$. Choosing
$t=\Theta(k\log n)$ ensures that $\|x-x'\|\le 1/n$. It is thus enough
to compute $x'_u-x'_v$.

Consider $x'_i=\sum_{i=0}^{t} e_u^\tran (\tfrac{1}{d}A)^i\cdot
(\tfrac{1}{d}b)$. Note that each term corresponds to a random walk
of length $t$ (using matrix $A$).

\begin{claim}
There is some $c>1$ and $i_0$, such that for $i\ge i_0$, and any node $u$
in $G$, the following holds. In the balanced case:
$$
\Pr_b\left[\left|e_u^\tran (\tfrac{1}{d}A)^i\cdot b\right|\ge O(\sqrt{\log
    n})\cdot(c^{-i/2}+n^{-1/2})\right]\le 1/n^2.
$$
In the unbalanced case:
$$
\Pr_b\left[\left|e_u^\tran (\tfrac{1}{d}A)^i\cdot b\right|\ge p+O(\sqrt{\log
    n})\cdot(c^{-i/2}+n^{-1/2})\right]\le 1/n^2.
$$
\end{claim}
\begin{proof}
Note that $(\tfrac{1}{d}A)^i$ corresponds to the following random walk
on $G$. We index the vertices of $G$ as $(v,q)$ where $v\in V_X$ and
$q\in\{1,2\}$ depending whether it is in $X_1$ or $X_2$. Then the
random walk is equivalent to: with probability 1/4 we take a self-loop
or jump into the $X_{2-q}$ component, and with probability $3/4$ we
take a random step in the $X_q$ component. Thus, we can consider a
graph $X'$ to be the graph $X$ where each node has a self-loop, and
thus degree $d=4$. The variable $v$ does a random walk in $X'$,
independently of $q$, whereas $q$ does a more complex walk (depending
on $v$). For the graph $X$, for any starting vertex, the probability
that the random walk hits a particular node $l$ in $X'$ is at most
$\alpha=c^{-i}+ 1/n$. This is easy to note by observing that if
$\tfrac{1}{d}A_X=\sum_j \lambda_j u_ju_j^\tran $ is the spectral
decomposition of the random walk matrix corresponding to $X$ with
$\lambda_1=1$ and $c^{-1}\triangleq\max_{j\ge 2} |\lambda_j|<1-\Omega(1)$, then
$\|(\tfrac{1}{d}A_X)^ie_v-\tfrac{1}{n}{\allones}\|_\infty\le
\|(\tfrac{1}{d}A_X)^ie_v-\tfrac{1}{n}{\allones}\|_2= \|\sum_{j\ge2} u_j\cdot
\lambda_j^i\cdot \langle u_j,e_v\rangle\|_2\le \max_{j\ge2} |\lambda_j^i|=c^{-i}$.

Let $\pi_{(l,q)}$ be the probability that the random walk of length
$i$, starting at $u$, stops at vertex $(l,q)$. By the above analogy
between random walk in $A$ vs random walk in $X'$, we have that
$\pi_{(l,q)}\le 2\alpha$. 

We now use Bernstein inequality to argue about the concentration of
$\sum_{(l,q)} \pi_{(l,q)}b_{(l,q)}$. This is where the balanced and
unbalanced case will differ. In the balanced case $b_{(l,q)}$ are
essentially random $\pm 1$, although there's a minor dependence: each
side is precisely balanced to 0. Hence we apply the inequality for
each of the 4 parts $P_1\ldots P_4$ of the graph. In each of the
parts, the values $b_{(l,q)}$ are independent. Hence, over the 4
parts, we have that, for any $z>0$:
$$
\Pr_b\left[\left|\sum_{(l,q)} \pi_{(l,q)}b_{(l,q)}\right|\ge 4z\right]\le
\exp\left[\tfrac{-z^2/2}{\max_{(l,q)} \pi_{(l,q)}\cdot (1+z/3)}\right].
$$ 

We take $z^2=2\ln n\cdot 3\cdot 2\alpha$, and then the probability
becomes $\le 1/n^2$. In particular, we have that $z=O(\sqrt{\log
  n\cdot \alpha})=O(\sqrt{\log n}\cdot (c^{-i/2}+n^{-1/2}))$, and
the claim follows for the balanced case.

In the unbalanced case, we can also consider the 4 parts $P_1,\ldots
P_4$. For each part $j$, we define the random variable, depending on
the bias the values of $b$. Wlog, suppose the part has a bias $+p$
(i.e., $b_{(l,q)}=+1$ with probability $1/2+p/2$). Then define
$Y_{(l,q)}=\pi_{(l,q)}b_{(l,q)}-\pi_{(l,q)}\cdot p$. Note that
$\E[Y_{(l,q)}]=0$ and
$\E[[Y_{(l,q)}^2]=\pi_{(l,q)}^2-(p\cdot\pi_{(l,q)})^2=\pi_{(l,q)}^2\cdot
  (1-p^2)\le \pi_{(l,q)}^2$. Applying Bernstein's inequality similarly
  to before, we have that
$$
\Pr_b\left[\left|\sum_{(l,q)} \pi_{(l,q)}b_{(l,q)}\right|\ge 4p+4z\right]\le
\exp\left[\tfrac{-z^2/2}{\max_{(l,q)} \pi_{(l,q)}\cdot (1+z/3)}\right].
$$ 
There's $z=O(\sqrt{\log n\cdot \alpha})$ making the above probability
$\le 1/n^2$. This finishes the unbalanced case.
\end{proof}

We now complete the proof of the lemma, for the balanced and
unbalanced cases. In the balanced case, we use the claim to
conclude that, by union bound, with probability at least $1-1/n$,
$$
|x'_u|\le \sum_{i=0}^t O(\sqrt{\log
  n})\cdot(c^{-i/2}+n^{-1/2})\le O(\sqrt{\log n}\cdot (1+t/n^{1/2})).
$$

Since the same bound hold for $x'_v$, and since $t\le O(k\log n)$, we
have that $|x'_u-x'_v|\le O(\sqrt{\log n}\cdot (1+k\log
n/n^{1/2}))\le O(\sqrt{\log n})$ for $k<O(n^{1/4}/\log n)$.

For the unbalanced case, we have, by union bound, with probability at
least $1-1/n$:
$$
|x'_u|\le \sum_{i=0}^t O(p+\sqrt{\log
  n})\cdot(c^{-i/2}+n^{-1/2})\le O(tp+\sqrt{\log n}).
$$
Replacing $t=O(k\log n)$ completes the unbalanced case.
\end{proof}

%% file: sddSolver.tex
\section{An SDD Solver}
\label{sec:SDD}

In this section we prove the following theorem
for solving linear systems in SDD matrices. 
To generalize from Laplacianss of regular graphs to SDD matrices, 
we face several issues as described in Section~\ref{subsec:techniques}.
We use the notation defined in \eqref{eq:tildeS}-\eqref{eqn:normalizedSsol}. 

\begin{thm}[SDD Solver]
\label{thm:SDD_solver_intro}
There exists a randomized algorithm,
that given input $\big\langle S,b,u,\epsilon,\tilde{\lambda}_{up}\big\rangle$,
where
\begin{itemize} \compactify
\item $S\in\mathbb{R}^{n\times n}$ is an SDD matrix,
\item $b\in\mathbb{R}^{n}$ is in the range of $S$ (equivalently, orthogonal
to the kernel of $S$),
\item $u\in[n]$, $\epsilon>0$, and
\item $\barkappa\ge1$ is an upper bound on the condition number $\kappa(\tilde S)$, 
\end{itemize}
this algorithm outputs $\hat{x}_{u}\in\mathbb{R}$ with the following guarantee.
Suppose $x^*$ is the solution for $Sx=b$ given in~\eqref{eqn:normalizedSsol}, 
then 
\[
  \forall u\in[n], \qquad
  \Pr\Big[|\hat{x}_{u}-x^*_u|\leq\epsilon||x^*||_{\infty}\Big]
  \geq 1-\frac{1}{s} 
\]
for suitable 
$s=O(\barkappa\log(\epsilon^{-1} \barkappa ||b||_{0} \cdot \frac{\max_{i\in[n]}D_{ii}}{\min_{i\in[n]}D_{ii}}))$.
The algorithm runs in time $O(f\epsilon^{-2}s^{3}\log s)$, 
where $f$ is the time to make a step 
in a random walk in the weighted graph formed by the non-zeros of $S$.
\end{thm}

Given an SDD matrix $S\in\mathbb{R}^{n\times n}$,
we may assume that $S_{ii}>0$ for every $i$ (as otherwise
the entire $i$-th row and column are zero and can be safely ignored).
Recall $D = \diag(S_{11},\ldots,S_{nn})$, and define $A\eqdef D-S$. 
Let $\tilde{A}\eqdef D^{-1/2}AD^{-1/2}$
(for intuition, this is the normalized adjacency matrix when $S$
is a Laplacian) and recall $\tilde{S} = D^{-1/2}SD^{-1/2}=I-\tilde{A}$
(the normalized Laplacian, respectively). For an eigenvalue $\mu$
of $\tilde{A}$, let $E_{\mu}(\tilde{A})\eqdef\{x:\tilde{A}x=\mu x\}$.
Observe that $\tilde{A}\preceq I\iff A\preceq D\iff S\succeq0$; 
recalling that $S$ is SDD, we conclude that $\tilde{A}\preceq I$. Moreover,
\[
  \tilde{A}x=x
  \iff D^{-1/2}(D-S)D^{-1/2}x=x
  \iff D^{-1/2}SD^{-1/2}x=0
  \iff SD^{-1/2}x=0,
\]
so $E_{1}(\tilde{A})=D^{1/2}\cdot\mbox{ker}(S)$. Observe that $\tilde{A}\succeq-I\iff A\succeq-D\iff D+A\succeq0$;
Since $D+A$ is SDD, we conclude that $\tilde{A}\succeq-I$. 
Let $1\geq\tilde{\lambda}_{1}\geq\tilde{\lambda}_{2}\geq \cdots \geq\tilde{\lambda}_{n}\geq-1$
be $\tilde{A}$'s eigenvalues with associated orthonormal eigenvectors
$\tilde{u}_{1},\ldots,\tilde{u}_{n}$ (note that $\tilde{A}$ is symmetric).
We can write $\tilde{A}=\tilde{U}\tilde{\Lambda}\tilde{U}^\tran $ where
$\tilde{U}=[\tilde{u}_{1}\ \tilde{u}_{2}\ldots\tilde{u}_{n}]$ is unitary
and $\tilde{\Lambda}=\diag(\tilde{\lambda}_{1},\ldots,\tilde{\lambda}_{n})$.
Note that $\tilde{S}=\tilde{U}(I-\tilde{\Lambda})\tilde{U}^\tran $,
and that $\tilde{S}^{+}=\tilde{U}(\frac{1}{1-\tilde{\lambda}_{1}},\ldots,\frac{1}{1-\tilde{\lambda}_{n}})\tilde{U}^\tran $
where by convention $\frac{1}{0}$ stands for $0$.
Let $d_{i}\eqdef D_{ii}$,
$d_{max}\eqdef\max_{i\in[n]}d_{i}$, and $d_{min}\eqdef\min_{i\in[n]}d_{i}$.
Let $\tilde{B}\eqdef\frac{\tilde{A}+I}{2}$. Note that $\tilde{B}=\tilde{U}\frac{\tilde{\Lambda}+I}{2}\tilde{U}^\tran $, and that the eigenvalues of $\tilde{B}$ are in $[0,1]$.
Let $\tilde{\lambda}\eqdef\max\{\frac{\tilde{\lambda}_{i}+1}{2}:i\in[n],\tilde{\lambda}_{i}<1\}$ (the largest non-one eigenvalue of $\tilde{B}$).
We now describe an algorithm that on input $b\in\mathbb{R}^{n}$ that
is in the range of $S$ (equivalently, is orthogonal to the kernel
of $S$), and $u\in[n]$, returns an approximation $\hat{x}_{u}$
to $x^*_u$, where $x^*=D^{-1/2}\tilde{S}^+D^{-1/2}b$ is the solution for $Sx=b$ given in \eqref{eqn:normalizedSsol}. %

\begin{algorithm}
\caption{$\hat{x}_{u}=\mbox{SolveLinearSDD}(S,b,||b||_{0},u,\epsilon,\tilde{\lambda})$}
\label{alg:SDDsolver}

\begin{enumerate} \compactify
\item Set
$
s=\log_{1/\tilde{\lambda}}(2\epsilon^{-1}(1-\tilde{\lambda})^{-1}\sqrt{||b||_{0}}\cdot\sqrt{{d_{max}}/{d_{min}}}),
$
and $\ell=O((\frac{\epsilon}{2s})^{-2}\log s)$.
\item For $t=0,1,\ldots,s-1$ do

\begin{enumerate} \compactify
\item Perform $\ell$ independent (lazy) random walks of length $t$ starting at
$u$, where in one step from vertex $v$, 
the walk stays put with probability $\frac{1}{2}$,
moves to $v'$ with probability $\frac{|A_{v v'}|}{2d_{v}}$, 
and terminates with the remaining probability.

For each walk $i\in[\ell]$, 
let $u_{i}^{(t)}$ be the vertex where the walk ended, 
and let $\sigma_{i}^{(t)}$ be the product of the signs along the walk
where stay-put steps have sign $1$ and others have $\sgn(A_{v v'})$. 
Formally, if the walk consists of $u=u_{0},u_{1},..,u_{t}$
then $\sigma_{i}^{(t)}=\prod_{j\in[t]}\sgn((D^{-1}A+I)_{u_{j-1},u_{j}})$
and if it terminated earlier then $\sigma_{i}^{(t)}=0$.

\item Set $\hat{x}_{u}^{(t)}=\frac{1}{\ell}\sum_{i\in[\ell]}
\sigma_{i}^{(t)}\frac{b_{u_{i}^{(t)}}}{d_{u_{i}^{(t)}}}$.
\end{enumerate}
\item Return $\hat{x}_{u}=\frac{1}{2}\sum_{t=0}^{s-1}\hat{x}_{u}^{(t)}$.
\end{enumerate}
\end{algorithm}

We now prove that Algorithm~\ref{alg:SDDsolver} indeed provides a good approximation.
Note that $b$ is orthogonal to $\mbox{ker}(S)$ iff
$D^{-1/2}b$ is orthogonal to $D^{1/2}\cdot\mbox{ker}(S)
=E_{1}(\tilde{A})=E_{1}(\tilde{B})$.

\begin{claim}
\label{claim:SDD_tail_bound}For $b$ that is orthogonal to the kernel
of $S$, 
\begin{equation}
\Big|x^*_u-\frac{1}{2}e_{u}^\tran D^{-1/2}\sum_{t=0}^{s-1}\tilde{B}^{t}D^{-1/2}b\Big|\leq\frac{\epsilon}{4}||D^{-1}b||_{\infty}.\label{eq:1}
\end{equation}
\end{claim}

\begin{proof}
Observe that
\[
  2\tilde{S}^+
  =\Big(\frac{I-\tilde{A}}{2}\Big)^+
  =\Big(I-\frac{\tilde{A}+I}{2}\Big)^+
  =(I-\tilde{B})^+.
\]
Thus, as $D^{-1/2}b$
is in the span of eigenvectors of $\tilde{B}$ with associated eigenvalues
in $[0,1)$, using the same idea as in Fact \ref{fact:inverse_power_series}
we get that 
\[
\tilde{S}^{+}D^{-1/2}b=\frac{1}{2}\sum_{t=0}^{\infty}\tilde{B}^{t}D^{-1/2}b ,
\]
and hence (recall $x^*_u=e_u^\tran D^{-1/2}\tilde{S}^{+}D^{-1/2}b$)
\[
x^*_u-\frac{1}{2}e_{u}^\tran D^{-1/2}\sum_{t=0}^{s-1}\tilde{B}^{t}D^{-1/2}b
=\frac{1}{2}e_{u}^\tran D^{-1/2}\sum_{t=s}^{\infty}\tilde{B}^{t}D^{-1/2}b.
\]
Similarly to the proof of Claim \ref{claim:Laplacian_tail_bound}, 
we now get
that 
\begin{eqnarray*}
|\frac{1}{2}e_{u}^\tran D^{-1/2}\sum_{t=s}^{\infty}\tilde{B}^{t}D^{-1/2}b| & \leq & \frac{1}{2\sqrt{d_{u}}}||\sum_{t=s}^{\infty}\tilde{B}^{t}D^{-1/2}b||_{2}\\
 &  & \leq\frac{1}{2\sqrt{d_{min}}}\sum_{t=s}^{\infty}\tilde{\lambda}^{t}||D^{-1/2}b||_{2}\\
 &  & \leq\frac{1}{2\sqrt{d_{min}}}\cdot\frac{\tilde{\lambda}^{s}}{1-\tilde{\lambda}}||D^{-1/2}b||_{2}\\
 &  & \leq\frac{1}{2}\sqrt{\frac{d_{max}}{d_{min}}}\cdot\frac{\tilde{\lambda}^{s}}{1-\tilde{\lambda}}||D^{-1}b||_{2}\\
 &  & \leq\frac{1}{2}\sqrt{\frac{d_{max}}{d_{min}}}\cdot\frac{\tilde{\lambda}^{s}}{1-\tilde{\lambda}}\sqrt{||b||_{0}}\cdot||D^{-1}b||_{\infty}=\frac{\epsilon}{4}||D^{-1}b||_{\infty}.
\end{eqnarray*}

\end{proof}

\begin{claim}
\label{claim:SDD_power_approximation}With probability at least $1-\frac{1}{s}$,
\[
\Big|\hat{x}_{u}-\frac{1}{2}e_{u}^\tran D^{-1/2}\sum_{t=0}^{s-1}\tilde{B}^{t}D^{-1/2}b\Big|\leq\frac{\epsilon}{4}||D^{-1}b||_{\infty}.
\]
\end{claim}

\begin{proof}
Recalling that $\tilde{A}=D^{-1/2}AD^{-1/2}$, we can write
\[
  D^{-1/2}\tilde{B}^{t} D^{-1/2}
  = D^{-1/2} \Big(\frac{\tilde A+I}{2}\Big)^t D^{-1/2}
  = D^{-1/2}\Big(D^{1/2} \frac{D^{-1}A+I}{2} D^{-1/2}\Big)^t D^{-1/2}
  =\Big(\frac{D^{-1}A+I}{2}\Big)^tD^{-1}.
\]
Hence (by induction), for every $t\in\{0,1,\ldots,s-1\}$ and $i\in[\ell]$,
we have 
\[
  e_{u}^\tran D^{-1/2}\tilde{B}^{t}D^{-1/2}b
  =e_u^\tran \Big(\frac{D^{-1}A+I}{2}\Big)^tD^{-1}b
  =\mathbb{E}[\sigma_{i}^{(t)}\frac{b_{u_{i}^{(t)}}}{d_{u_{i}^{(t)}}}].
\]
By a union bound over Hoeffding bounds (as $|\sigma_{i}^{(t)}\frac{b_{u_{i}^{(t)}}}{d_{u_{i}^{(t)}}}|\leq||D^{-1}b||_{\infty}$),
with probability at least $1-\frac{1}{s}$, for every $t\in\{0,1,\ldots,s-1\}$,
\[
  \Big|\frac{1}{\ell}\sum_{i\in[\ell]}
    \sigma_{i}^{(t)}\frac{b_{u_{i}^{(t)}}}{d_{u_{i}^{(t)}}}
      -e_{u}^\tran D^{-1/2}\tilde{B}^{t}D^{-1/2}b
    \Big|
  \leq
  \frac{\epsilon}{2s}||D^{-1}b||_{\infty},
\]
which implies that 
\[
  \Big|\frac{1}{2}\sum_{t=0}^{s-1}\hat{x}_{u}^{(t)}-
  \frac{1}{2}e_{u}^\tran D^{-1/2}\sum_{t=0}^{s-1}\tilde{B}^{t}D^{-1/2}b
  \Big|
  \leq
  \frac{\epsilon}{4}||D^{-1}b||_{\infty}.
\]

\end{proof}

Combining Claim \ref{claim:SDD_tail_bound} and Claim \ref{claim:SDD_power_approximation}
we get that (with probability $1-\frac{1}{s}$)
$|\hat{x}_{u}-x^*_u|\leq\frac{\epsilon}{2}||D^{-1}b||_{\infty}$.
Now, letting $x$ denote any solution to the system $Sx=b$,
for every $i\in[n]$ we have
\[
\frac{|b_{i}|}{d_{i}}=\frac{|\sum_{j\in[n]}S_{ij}x_{j}|}{d_{i}}\leq\frac{\sum_{j\in[n]}|S_{ij}|\cdot||x||_{\infty}}{d_{i}}\leq\frac{2d_{i}||x||_{\infty}}{d_{i}}=2||x||_{\infty}
\]
where the last inequality is because $S$ is SDD. Therefore, $||D^{-1}b||_{\infty}\leq2||x||_{\infty}$,
and we conclude that (with probability $1-\frac{1}{s}$)
$|\hat{x}_{u}-x_u^*|\leq\epsilon||x||_{\infty}$
for every solution $x$ to the system $Sx=b$ (and in particular for $x^*$).
We now turn to the runtime
of Algorithm~\ref{alg:SDDsolver}, which is dominated by the time it takes to perform
the random walks. There are $s\cdot\ell$ random walks in total. Let
$f$ be the time it takes to make a single step in the random walks
of Algorithm~\ref{alg:SDDsolver} (it depends on the access method/representation of
$S$ and/or its sparsity). The random walks do not need
to be independent for different values of $t$ (as we applied a union
bound over the different $t$), we can extend, at each iteration
$t$, the $\ell$ respective random walks constructed at iteration
$t-1$ by an extra step in time $f$, obtaining a total runtime
$O(s\cdot\ell\cdot f)=O(f\epsilon^{-2}s^{3}\log s)$.
We conclude the following.

\begin{thm} \label{thm:SDD_solver}
Given access to an SDD matrix $S\in\mathbb{R}^{n\times n}$, $b\in\mathbb{R}^{n}$
that is orthogonal to the kernel of $S$, $||b||_{0}$, $u\in[n]$,
$\epsilon>0$, and $\tilde{\lambda}=\max\{\frac{\tilde{\lambda}_{i}+1}{2}:i\in[n],\tilde{\lambda}_{i}<1\}$,
with probability at least $1-\frac{1}{s}$, Algorithm~\ref{alg:SDDsolver} outputs a
value $\hat{x}_{u}\in\mathbb{R}$ such that
$|\hat{x}_{u}-x_u^*|\leq\epsilon||x||_{\infty}$
for every solution $x$ to the system $Sx=b$ (and in particular for $x^*$).
Algorithm~\ref{alg:SDDsolver} runs in time $O(f\epsilon^{-2}s^{3}\log s)$
where $f$ is the worst-case time to make a step in a random
walk in the weighted graph formed by the non-zeros of $S$, and $s=O(\log_{1/\tilde{\lambda}}(\epsilon^{-1}(1-\tilde{\lambda})^{-1}||b||_{0}\cdot\frac{d_{max}}{d_{min}}))$.
\end{thm}

\begin{proof}[Proof of Theorem~\ref{thm:SDD_solver_intro}]
Recall $1\geq\tilde{\lambda}_{1}\geq\tilde{\lambda}_{2}\geq \cdots \geq\tilde{\lambda}_{n}\geq-1$
are the eigenvalues of $\tilde{A} = I-\tilde{S}$,
and that 
$\tilde{\lambda}=\max\{\frac{\tilde{\lambda}_{i}+1}{2}:i\in[n],\tilde{\lambda}_{i}<1\}$. 
Then the smallest non-zero eigenvalue of $\tilde S$ is
$1-\tilde{\lambda}_{i} = 1-(2\tilde\lambda-1) = 2(1-\tilde\lambda)>0$.
The largest eigenvalue of $\tilde S$ is $1-\tilde{\lambda}_{n}\in [1,2]$
where the lower bound is by the following argument 
(since all diagonal entries in $\tilde S$ are $1$, 
and the off-diagonal entries contribute 0 in expectation)
\[
  \max_{x\neq 0} \frac{x^\tran  \tilde S x}{x^\tran  x}
  \ge \EX_{x\in\set{\pm1}^n} \Big[ \frac{x^\tran  \tilde S x}{n} \Big]
  = 1 .
\]
Thus, $\kappa(\tilde S) = \Theta(\frac{1}{1-\tilde\lambda})$.
Using Fact~\ref{fact:inequality_e} with $\delta=1-\tilde{\lambda}$,
the expression above for $s$ becomes 
$s=O(\kappa(\tilde S) \log(\epsilon^{-1} \kappa(\tilde S) ||b||_{0} \cdot \frac{d_{max}}{d_{min}}))$.

The theorem follows by noting that the analysis
in Claims~\ref{claim:SDD_tail_bound} and~\ref{claim:SDD_power_approximation}
holds also when replacing $||b||_{0}$ by an upper bound $B_{up}\geq||b||_{0}$
and $\kappa(\tilde S)$ by an upper bound $\barkappa$
(or equivalently $\tilde{\lambda}$ by an upper bound $\tilde{\lambda}_{up}$). 
\end{proof}